\documentclass[a4paper,11pt]{article}
\pdfoutput=1
\RequirePackage{ifluatex}
\usepackage{fullpage}
\usepackage{array}
\usepackage{xspace}
\usepackage[utf8]{inputenc}
\usepackage[T1]{fontenc}
\usepackage{graphicx}
\usepackage{amsmath,bm}
\usepackage{float}
\usepackage{amsfonts}
\usepackage{algorithm}      
\usepackage[noend]{algpseudocode}  
\usepackage{algorithmicx}
\usepackage[dvipsnames,usenames]{color}
\usepackage[colorlinks=true,urlcolor=Blue,citecolor=Green,linkcolor=BrickRed]{hyperref}
\usepackage[usenames,dvipsnames]{xcolor}
\usepackage{authblk}
\usepackage{amsthm}
\usepackage{bigstrut}
\usepackage[noadjust]{cite}
\usepackage{todonotes}
\usepackage{thmtools,thm-restate}
\usepackage{lineno}
\usepackage{comment}

\newcommand{\Oh}{\mathcal{O}}
\newcommand{\degree}{\mathsf{deg}}

\newcommand{\nodelabel}{\ell}

\newcommand{\heavy}{\mathsf{heavy}}
\newcommand{\head}{\mathsf{head}}
\newcommand{\start}{\mathsf{start}}
\newcommand{\bound}{\mathsf{bound}}
\newcommand{\interval}{\mathsf{I}}
\newcommand{\nodespan}{\mathsf{span}}
\newcommand{\treeroot}{\mathsf{root}}
\newcommand{\level}{\mathsf{level}}
\newcommand{\seglen}{\mathsf{sl}}
\newcommand{\seglenp}{\mathsf{sl'}}
\newcommand{\seglenb}{\mathsf{sl''}}
\newcommand{\rt}{\mathsf{rt}}
\newcommand{\lightweight}{\mathsf{lw}}
\newcommand{\lt}{\mathsf{local}}
\newcommand{\sma}{\mathsf{small}}
\newcommand{\smap}{\mathsf{small'}}
\newcommand{\rank}{\mathsf{rank}}

\newenvironment{proofs}{%
  \proof}{\endproof}

\usepackage{tabularx} 
\newcolumntype{C}{>{\centering\arraybackslash}X} 

\newtheorem{theorem}{Theorem}
\newtheorem{lemma}[theorem]{Lemma}
\newtheorem{proposition}[theorem]{Proposition}
\newtheorem{claim}[theorem]{Claim}
\newtheorem{fact}[theorem]{Fact}
\newtheorem{definition}[theorem]{Definition}
\newtheorem{corollary}[theorem]{Corollary}

\title{Shorter Labels for Routing in Trees}

\author{Pawe{\l} Gawrychowski}
\author{Wojciech Janczewski}
\author{Jakub Łopuszański}
\affil{University of Wrocław, Poland}

\date{}

\begin{document}

\maketitle

\begin{abstract}
A routing labeling scheme assigns a binary string, called a label, to each
node in a network, and chooses a distinct port number from $\{1,\ldots,d\}$
for every edge outgoing from a node of degree $d$.
Then, given the labels of $u$ and $w$ and no other information about
the network, it should be possible to determine the port number corresponding
to the first edge on the shortest path from $u$ to $w$.
In their seminal paper, Thorup and Zwick [SPAA 2001] designed several
routing methods for general weighted networks.
An important technical ingredient in their paper that according
to the authors ``may be of independent practical and theoretical interest''
is a routing labeling scheme for trees of arbitrary degrees.
For a tree on $n$ nodes, their scheme constructs labels consisting
of $(1+o(1))\log n$ bits such that the sought port number can be
computed in constant time. Looking closer at their construction,
the labels consist of $\log n + \Oh(\log n\cdot \log\log\log n / \log\log n)$ bits.
Given that the only known lower bound is $\log n+\Omega(\log\log n)$, a natural
question that has been asked for other labeling problems in trees is to determine
the asymptotics of the smaller-order term.

We make the first (and significant) progress in 19 years on determining the correct
second-order term for the length of a label in a routing labeling scheme for trees on $n$ nodes.
We design such a scheme with labels of length $\log n+\Oh((\log\log n)^{2})$.
Furthermore, we modify the scheme to allow for computing the port number
in constant time at the expense of slightly increasing the length to $\log n+\Oh((\log\log n)^{3})$.
\end{abstract}

\thispagestyle{empty}
\clearpage
\setcounter{page}{1}

\section{Introduction}

The notion of informative labeling schemes, formally introduced by Peleg~\cite{Peleg05}, elegantly
captures the idea that, for a very large and scattered network which needs to be accessed from
different locations, it is desirable to select the identifiers of the nodes as to encode some information
about the underlying topology. Of course, to make this notion useful one needs to restrict the length
of each identifier, as otherwise it would be possible to simply store a description of the whole network there.
A particularly natural example considered by Kannan et al.~\cite{Kannan} is assigning the identifiers
in such a way that one can decide if two nodes are adjacent by inspecting their identifiers, without using
any additional information about the network. This can be seen as pushing the distributed aspect
of representing a network to the extreme. 

An informative labeling scheme assigns a short binary string, called a label, to each node in a network,
so that a function defined on subsets of nodes can be calculated for any subset by inspecting the labels
of the nodes. Formally, such a scheme consists of an encoder that assigns the labels given a description
of the whole network, and a decoder that evaluates the function on a subset given the labels of the
nodes. The primary goal is to minimise the maximum length of a label assigned to a node,
while efficient implementations of the decoder and the encoder are the secondary goals. 
We stress that the only information available to the decoder are the given labels, and it cannot
access any other information about the network. This question has been considered for functions such
as adjacency~\cite{Kannan,alstrup2015optimal,petersen2015near,alstrup2015adjacency,AlonN17,bonichon2007short}, distance~\cite{AbboudGMW17,GawrychowskiU16,GawrychowskiKU16,alstrup2015distance,alstrup2005labeling,gavoille2004distance,alstrup2016simpler,FGNW16,GavoilleKKPP01,gavoille2007local,peleg2000proximity,kaplan2001short}, flows and connectivity~\cite{KatzKKP04,HsuL09,Korman10} or Steiner tree~\cite{Peleg05}.
See~\cite{rotbart2016new} for an up-to-date survey.
Interestingly, as observed by Kannan et al.~\cite{Kannan}, the question of designing an adjacency labeling
scheme for a given class of graphs is in fact equivalent to constructing the so-called induced universal graph,
that is, a larger graph containing each graph from the class as a node-induced subgraph. This purely
combinatorial problem has been been already studied in 60s~\cite{moon_1965}.

\paragraph{Routing.}
A fundamental challenge related to networks is that of designing efficient routing mechanisms that
are able to transfer information between any two nodes of the network along the shortest (or, at least,
a reasonably short) path. A routing scheme formalises this as follows. Each node of the network
receive packets of information and needs to decide whether they have already reached their destination or should
be forwarded to one of its neighbours. This decision is made based on a header attached to the
packet and a local routing table. The edges outgoing from each node have their associated port numbers
and, if the decision is to forward the packet, the node needs to determine a port number and forward
the packet along the corresponding edge. The stretch of a routing scheme is the maximum ratio
between the length of a path determined by the scheme and the length of the shortest path. In most
routing schemes the header is simply the label of the destination node chosen by the designer of the network.

In this model, Cowen designed a stretch-3 routing scheme with headers of size $\Oh(\log n)$ and local
routing tables of size $\tilde\Oh(n^{2/3})$~\cite{Cowen01}. Eilam et al.~\cite{EilamGP03} decreased the
size of the local routing tables to $\tilde\Oh(n^{1/2})$ at the expense of increasing the stretch to 5.
Then, in their seminal paper Thorup and Zwick~\cite{thorup2001compact} designed a stretch-3 routing scheme with headers of size
$(1+o(1))\log n$ and routing tables of size $\tilde\Oh(n^{1/2})$ (essentially optimal for such a stretch).
An important technical ingredient introduced in their paper is a routing scheme for trees of unbounded
degree that assigns a label of length $(1+o(1))\log n$ bits to every node in such a way that, given only the label
of a source node and the label of a destination node, it is possible to determine in constant time
the port number of the first edge on the unique path from the source to the destination. According
to the authors, this ingredient ``may be of independent practical and theoretical interests''. Indeed,
follow-up works on this topic also reduce routing in a general graph to routing in (multiple) trees.

\paragraph{Labeling schemes in trees.}
Arguably, trees are one of the most important classes of graphs considered in the context of labeling schemes.
Functions studied in the literature on labeling schemes in trees include
adjacency~\cite{alstrup2015optimal,alstrup2002small,bonichon2007short},
ancestry~\cite{abiteboul2006compact,fraigniaud2010compact,FraigniaudAncestry},
routing~\cite{FraigniaudG01,FraigniaudG02,thorup2001compact},
distance~\cite{peleg2000proximity,gavoille2004distance,alstrup2005labeling,alstrup2015distance,gavoille2007distributed},
and nearest common ancestors~\cite{Peleg05,AlstrupGKR02,fischer2009short,AHL14,GawrychowskiL17}.
See Table~\ref{tab:treesbest} for a summary of the state-of-the-art bounds for these problems.

\begin{table}
\centering
 \begin{tabular}{|c|c|c|} 
 \hline
 \textbf{Function} & \textbf{Lower bound} & \textbf{Upper bound}\\ 
 \hline
 Distance & $1/4\log^2 n - \Oh(\log n)$ \cite{alstrup2015distance} & $1/4\log^2 n + o(\log^2 n)$ \cite{FGNW16} \\
 \hline
 Fixed-port routing & $\Omega(\log^2 n/\log{\log{n}})$ \cite{FraigniaudG02} & $\Oh(\log^2 n/\log{\log{n}})$ \cite{FraigniaudG01} \\
 \hline
 $(1+\epsilon)$-approximate distance & $\Omega(\log{(1/\epsilon)}\log{n})$ \cite{FGNW16} & $\Oh(\log{(1/\epsilon)}\log{n})$ \cite{FGNW16} \\
 \hline
 Nearest common ancestor & $1.008\log{n}-\Oh(1)$ \cite{AHL14} & $2.318\log{n}+\Oh(1)$ \cite{GawrychowskiL17} \\
 \hline
 Designer-port routing & $\log{n}+\Omega(\log{\log{n}})$ \cite{alstrup2005labeling} & $\pmb{\log{n}+\Oh((\log{\log{n}})^2)}$ \\
 \hline
 Ancestry & $\log{n}+\Omega(\log\log n)$~\cite{alstrup2005labeling} & $\log n+\Oh(\log\log n)$~\cite{FraigniaudAncestry} \\
 \hline
 Siblings/connectivity (unique) & $\log{n}+\Omega(\log{\log{n}})$ \cite{alstrup2005labeling} & $\log{n}+\Oh(\log{\log{n}})$ \cite{alstrup2005labeling} \\
 \hline
 Adjacency & $\log{n}$ (trivial) & $\log{n}+\Oh(1)$ \cite{alstrup2015optimal} \\
 \hline
\end{tabular}
\caption{Summary of the state-of-the-art bounds for labeling schemes in trees.}
\label{tab:treesbest}
\end{table}

For the first three problems,
it was not too hard to design schemes with labels of length
$\Oh(\log n)$, and lower bounds of $\log n$ are also quite straightforward.
In case of adjacency labeling, we can store the preorder numbers of the node and its parent
in the label, while for ancestry we can store the pre- and post-order numbers of the node.
For routing, \cite{FraigniaudG01} gives a fairly simple solution with labels of
length $7\log{n}$. Thus, the first challenge was to determine the exact coefficient in the
first-order terms. For all three problems, this turned out to be 1. This brought the necessity
of considering the second-order terms to differentiate the complexities of these problems.
After a series of intermediate results, Alstrup, Dahlgaard, and Knudsen~\cite{alstrup2015optimal}
presented an adjacency labeling scheme with labels of length $\log n+\Oh(1)$.
On the other hand, ancestry labeling is known to require $\log n+\Omega(\log\log n)$ bits~\cite{alstrup2005labeling},
and schemes with labels of length $\log n+\Oh(\log\log n)$ are known~\cite{FraigniaudAncestry}.
Thus, we have a separation between adjacency and ancestry labeling in trees.
For routing, the ancestry lower bound of $\log n+\Omega(\log\log n)$ essentially carries over,
but from the upper bound perspective there was no progress after Thorup and Zwick~\cite{thorup2001compact}
presented their scheme with labels of length $\log n+\Oh(\log n \cdot \log\log\log n/\log\log n)$.
Table~\ref{tab:treeshistory} summarizes the subsequent improvements for all three mentioned problems.

\begin{table}
\begin{minipage}{0.48\textwidth}
 \begin{tabularx}{6.5cm}{|C|} 
 \hline
 \textbf{Adjacency in trees} \\
 \hline 
 $2\log{n}$ (trivial) \\
 \hline
 $\log{n}+\Oh(\log{\log{n}})$ \cite{ChungUniversal}\\
 \hline 
 $\log{n}+\Oh(\log^{*}{n})$ \cite{alstrup2002small} \\
 \hline 
 $\log{n}+\Oh(1)$ \cite{alstrup2015optimal}\\
 \hline
 \end{tabularx}
 
 \medskip
 
 \begin{tabularx}{6.5cm}{|C|} 
 \hline
 \textbf{Routing in trees} \\
 \hline 
 $\Oh(\log{n})$ \cite{FraigniaudG01} \\
 \hline 
 $\log n+\Oh(\log n\cdot \log\log\log n / \log\log n)$ \cite{thorup2001compact} \\
 \hline
 $\pmb{\log{n}+\Oh((\log{\log{n}})^2)}$ \\
 \hline
\end{tabularx}
\end{minipage}
\begin{minipage}{0.48\textwidth}
 \begin{tabularx}{6.5cm}{|C|} 
 \hline
 \textbf{Ancestry in trees} \\
 \hline 
 $2\log{n}$ (trivial) \\
 \hline
 $1.5\log{n}+\Oh(\log{\log{n}})$ \cite{AbiteboulAncestor} \\
 \hline 
 $\log{n}+\Oh(\sqrt{\log{n}})$ \cite{AlstrupAncestorSqrt} \\
 \hline 
 $\log n+\Oh(\log n\cdot \log\log\log n / \log\log n)$ \cite{thorup2001compact} \\
 \hline 
 $\log{n}+4\log{\log{n}}+\Oh(1)$ \cite{FraigniaudAncestry} \\
 \hline 
 $\log{n}+2\log{\log{n}}+3$ \cite{DahlgaardKR15} \\
 \hline 
 \end{tabularx}
 \hfill
\end{minipage}
\caption{Progress on improving the smaller-order term for different labeling problems in trees.}
\label{tab:treeshistory}
\end{table}

\paragraph{Routing labeling schemes in trees.}
There are two models for routing labeling schemes in trees: fixed-port and designer-port. In both versions, each undirected edge is replaced
with two directed edges and, for every node $u$ of degree $\degree(u)$, the edges outgoing from $u$ have their assigned
port numbers that form a permutation of $\{1,\ldots,\degree(u)\}$. Given the labels of $u$ and $w$ we should output the port number
corresponding to the first edge on the path from $u$ to $w$. The difference between the models
is that in the fixed-port model the port numbers are predetermined and cannot be modified,
while in the designer-port model we are free to select them while assigning the labels, as long
as the port numbers assigned to the edges outgoing from $u$ form a permutation of
$\{1,\ldots,\degree(u)\}$. Intuitively, the additional degree of freedom could allow for shorter labels.
For trees on $n$ nodes it is known that labels of length $\Oh(\log^{2}n/\log\log n)$
are enough in the fixed-port model~\cite{FraigniaudG01}, and this is asymptotically tight~\cite{FraigniaudG02}.
In the designer-port model, Fraigniaud and Gavoille~\cite{FraigniaudG01} showed that labels
of length $4.8\log n$ are enough.
Then, Thorup and Zwick~\cite{thorup2001compact} achieved labels of length $\log n+\Oh(\log n\cdot \log\log\log n / \log\log n)$.
For the designer-port model the only known lower bound of  $\log n+\Omega(\log\log n)$ bits follows by
a simple reduction from ancestry labeling\footnote{\cite{rotbart2016new} explicitly mentions that such a lower bound holds,
but does not describe the (simple) necessary modification. For completeness, we provide a self-contained description in Section~\ref{sec:lower_bound}.},
for which an upper bound of $\log n+\Oh(\log\log n)$
does exist~\cite{FraigniaudAncestry}, thus such a reduction cannot result in a better lower bound.

\paragraph{Our results.}
In the designer-port model,
we construct a labeling scheme for routing in trees on $n$ nodes with labels of length
$\log n+\Oh((\log\log n)^{2})$ (Theorem~\ref{th:main}).
Thus, we make significant progress on determining the correct second-order term,
exponentially improving on the previous results.
Furthermore, our scheme is canonical, meaning that the ports are assigned in the natural order
corresponding to the sizes of the subtrees.
We also show how to extend our scheme so that the query can be answered in constant time,
assuming the standard word RAM model with words of logarithmic size,
at the expense of slightly increasing the length of the labels to
$\log n+\Oh((\log\log n)^{3})$ bits (Theorem~\ref{th:const}).
Additionally, we describe how our schemes work for trees of bounded degree or depth,
and present some trade-offs if local tables at nodes are allowed to be larger than single headers.

\paragraph{Overview of the methods.}
Similarly as in many other papers on labeling schemes in trees, we extensively use the idea of rounding up
numbers to a power of $2^{1/b}$.
For integer parameter $b$, any integer $x\in [0,n-1]$ can be rounded up to $x' = \lfloor 2^{t/b} \rfloor$,
for the smallest integer $t$ such that $x' \geq x$ holds.
Assuming that $b$ is known, only the value of $t$ needs to be stored using only $\log(b\log n)=\log b+\log\log n$ bits instead of $\log n$.
This is useful when combined with the notion of an interval-based scheme,
in which we assign a range $[\start(u),\start(u)+\bound(u))$ to every node, and then round up $\bound(u)$.

We follow the particularly clean version of such an approach
used by Dahlgaard, Knudsen, and Rotbart~\cite{DahlgaardKR15} in their ancestry labeling scheme.
This allows us to decide if we should route up, but the real challenge is to route down.
In Section~\ref{sec:preliminary}, we overcome
this difficulty by partitioning the children of $u$ into small and big. The intuition is that the former can be
rounded up more aggressively, and we can afford to store the number of children for
every possible value of $\bound$ corresponding to the latter. By appropriately adjusting the parameters
in this construction, we obtain a canonical routing labeling scheme with labels of length $\log n+\Oh((\log^{3/4}n)\sqrt{\log\log n})$.
This is already a significant improvement on the state-of-the-art, and serves as an introduction
to the better scheme presented in Section~\ref{sec:improved2}.
This section introduces our main technical
contribution, which is a compact encoding of the sizes of the subtrees rooted at the children.
Of course we cannot afford to store the sizes exactly, and hence proceed in two steps. Firstly, we round up
every size, with smaller subtrees being rounded more aggressively. Secondly, we partition the children
into groups and round up the size of every subtree in a group to the same value.
This results in a canonical routing labeling scheme with labels of length $\log n+\Oh(\sqrt{\log n\cdot \log\log n})$.

Finally, in Section~\ref{sec:not_so_final} we combine this with an additional tool used for adjacency labeling~\cite{alstrup2015optimal}:
we are able to guarantee that, for some nodes $u$, there are many trailing zeroes in $\start(u)$, which
allows us to hide some additional information there.
With that we obtain a canonical routing labeling scheme with labels of length $\log n+\Oh((\log\log n)^{2})$.
While we are mostly concerned with determining the asymptotics of the second-order term, we observe
that the decoder for this scheme can be implemented in polylogarithmic time. In Appendix~\ref{sec:const} 
we describe how to modify the labeling to allow answering queries in constant time at the expense
of increasing the length of the labels to $\log n+\Oh((\log\log n)^{3})$.

\section{Preliminaries}
\label{sec:preliminaries}

We consider rooted trees. For a rooted tree $T$, $|T|$ denotes its size (the number of nodes).
This will be denoted by $n$ for the whole input tree.
$\treeroot(T)$ denotes the root of $T$.
$\degree(u)$ denotes the degree (number of children) of $u$.
The subtree rooted at node $u\in T$ is denoted by $T_{u}$.
We define the level of a node $u$, denoted $\level(u)$, to be the unique integer $l$ such that $|T_{u}|\in [2^{l},2^{l+1})$.
Thus, $\level(u)\in \{0,1,\ldots,\log n\}$.  

We need the standard notion of heavy path decomposition of a tree $T$. Each
non-leaf node $u\in T$ selects its child $\heavy(u)$ with the largest subtree (if there is a tie,
we choose any largest subtree).
All other children are called light.
The edge from $u$ to $\heavy(u)$ is called heavy, and all other edges outgoing from $u$ are called light.
A heavy path is a maximal path $P$ consisting of heavy edges, with its rootmost node
denoted by $\head(P)$ (note that such a path must terminate at a leaf).
Heavy paths form a partition of the nodes of $T$.
The light depth of a node $u\in T$ is the number of light edges on the
path from $u$ to the root and is at most $\log n$~\cite{sleator1983data}.

If $s$ is a binary string, $|s|$ denotes its length.
For binary strings $s_1,s_2$ we denote their concatenation by $s_1 \circ s_2$.
$s^k$ is $k$ copies of $s$ concatenated together.
Our labeling schemes will use labels composed of constant number of parts.
The length of each part will be $\Oh(\log{n})$, and as our bounds for length of labels are $\log{n}+\Omega(\log{\log{n}})$,
we will organise them internally as follows.
Let $s_i$ be the $i$-th part of a label, and $l_i$ be a string storing binary representation of $|s_i|$
(so $|l_i| = \Oh(\log{\log{n}})$ in our case).
Then the label is stored as $0^{|l_1|} \circ 1 \circ l_1 \circ s_1 \circ
0^{|l_2|} \circ 1 \circ l_2 \circ s_2 \circ\ldots$, from which any part can be extracted by checking the position
of the most significant bit and then taking prefixes of indicated lengths.
In fact, one can find the most significant bit in constant time even without dedicated operation,
just with basic arithmetics, as presented in the paper of Fredman and Willard on fusion trees~\cite{FusionTrees}.

Let $\mathcal{T}$ be a family of rooted trees. A routing labeling scheme for $\mathcal{T}$ consists of an encoder and a decoder.
\begin{definition}
The encoder takes a tree $T\in \mathcal{T}$ and assigns a
label (binary string) $\nodelabel(u)$ to every node $u\in T$. Additionally, every edge from
$u$ to its child is labeled with a distinct integer from $\{1,2,\ldots,\degree(u)\}$, called
the port number.
\end{definition}
\begin{definition}
The decoder receives labels $\nodelabel(u)$ and
$\nodelabel(w)$, such that $u,w\in T$ for some $T\in \mathcal{T}$ and $u \neq w$.
If the next node on the path from $u$ to $w$ is the parent of $u$, the decoder should return 0\footnote{We might also allow assigning
a port number to the edge from $u$ to its parent. However, we find this formulation cleaner and sufficient for our purposes.}.
Otherwise, it should return the port number corresponding to the first edge
on the path.
\end{definition}
The decoder is not aware of $T$ and only knows that $u$
and $w$ come from the same tree belonging to $\mathcal{T}$.
Our schemes will always create unique labels, so in fact they do not need the assumption that $u\neq w$.
We are interested in minimising the maximum length of a label, that is, $\max_{T\in \mathcal{T}}\max_{u\in T} |\ell(u)|$.
For convenience, we will assume that the value of $\lceil \log{n} \rceil$ is known to the decoder, where $n$ is the size
of the tree. This is without loss of generality, as we can simply include it in every label.
To avoid clutter in the description of our schemes, we will sometimes ignore floors and ceilings.

\begin{definition}
(as in~\cite{thorup2001compact})
We say that the assignment of port numbers to the edges of $T$ is
\emph{canonical} if it is obtained in the following way. Let $v$ be a node of $T$ whose parent, if any, is $v_0$,
and whose children arranged in non-increasing order of size are $v_1, v_2, \ldots v_i$,
i.e. $|T_{v_1}| \geq |T_{v_2}| \geq \ldots \geq |T_{v_i}|$.
Then, for every $j$, the edge from $v$ to $v_{j}$ should be assigned port $j$.
There might be many canonical assignments if subtrees of some children are of equal sizes.
\end{definition}

\paragraph{Computational model. } We assume the standard word RAM model with words consisting
of $w=\Omega(\log n)$ bits. Basic arithmetical operations (addition, subtraction, bitwise operations, multiplication and division)
on such words are assumed to work in constant time. A label consisting of $b$ bits is assumed to be given to the decoder
packed in an array consisting of $\lceil b/w\rceil$ words.

\section{Framework for labeling schemes}
\label{sec:preliminary}
In this section, we first present the framework used by all of our labeling schemes,
based on the ancestry labeling scheme of Dahlgaard et al.~\cite{DahlgaardKR15}.
Then we describe how to extend it to routing for bounded degree trees,
and finally an introductory routing labeling scheme for general trees is presented,
already with better second-order term than the previously known results.

Let $b$ be an integer parameter to be chosen later.
We will assign an interval of integers $\interval(u)=[\start(u),\start(u)+\bound(u))$ to every node $u\in T$.
$\start(u)$ can be thought of as the unique ID of $u$ and we will sometimes refer to it as such.
We shall guarantee the following properties:
\begin{enumerate}
\item If $u$ and $w$ belong to different heavy paths, then $\interval(u)\cap\interval(w)=\emptyset$,
$\interval(u)\subset \interval(w)$, or $\interval(w)\subset \interval(u)$. Furthermore, 
$\interval(u)\subset \interval(w)$ iff $w$ is an ancestor of $u$.
\item If $u$ and $w$ belong to the same heavy path and $u$ is closer to the root of $T$ than
$w$, then $\start(w)\in \interval(u) \textbackslash \{\start(u)\}$. 
\item $\bound(u)=\lfloor 2^{t/b}\rfloor$ for some nonnegative integer $t$, depending on $u$.
\end{enumerate}

From the properties of intervals stated above we have that
$u$ is a proper ancestor of $w$ iff $\start(w)\in \interval(u) \textbackslash \{\start(u)\}$.
Additionally, let $\nodespan(u)=\max_{u'\in T_{u}}\start(u')+\bound(u')-\start(u)$.
To facilitate routing, we will need a specific
assignment with more properties, but first we are going to present a simpler version that is only able to answer ancestry queries.
The ancestry scheme described below is just a small adjustment of~\cite{DahlgaardKR15}.

\subsection{Labeling for ancestry queries}
The intervals are assigned in a bottom-up order.
Consider a heavy path $P=u_{1}-u_{2}-\ldots-u_{p}$, where $\head(P)=u_{1}$,
and assume that by recursion we have already properly assigned intervals to all nodes in the subtrees hanging off $P$,
but for now these intervals are not disjoint, violating property $1$.
We will appropriately shift all of them to guarantee
that the intervals of nodes belonging to different subtrees hanging off $P$ are disjoint.
This is readily done by scanning the path from top to bottom while maintaining
an accumulator initially set to 0. After reaching a node $u_{i}$,
we set $\start(u_{i})$ to be the current accumulator,
increase the accumulator by 1 and iterate over the light children $v_{i,1},v_{i,2},\ldots$
of $u_{i}$. For each such node $v_{i,j}$, we shift all intervals assigned to nodes of $T_{v_{i,j}}$ by
the current accumulator, just by adding the value of the accumulator to $\start$ values for the nodes in $T_{v_{i,j}}$,
and then increase the accumulator by $\nodespan(v_{i,j})$. Finally, we need
to adjust $\bound(u_{i})$ for every $i=1,2,\ldots,p$. It can be seen that we only need that
$\start(u_{i})+\bound(u_{i})$ is at least as large as the final accumulator. Furthermore, the final
accumulator is simply $A=\sum_{i=1}^{p}(1+\sum_{j\geq 1}\nodespan(v_{i,j}))$. Then, because the possible values
of $\bound(u_{i})$ are all numbers of the form $\lfloor 2^{t/b}\rfloor$, we can always
adjust $t$ so that $\start(u_{i})+\bound(u_{i}) \leq 2^{1/b}A$. Thus, we obtain that
$\nodespan(u_{1}) \leq 2^{1/b}\sum_{i=1}^{p}(1+\sum_{j\geq 1}\nodespan(v_{i,j}))$.
We will call this whole step for a single heavy path the shifting procedure.

By unwinding the recurrence and using the fact that the light depth of any node is at most $\log n$, we conclude that:
\[ \nodespan(\treeroot(T)) \leq 2^{(\log n)/b}n ,\]
and by construction $\start(u) + \bound(u) \leq \nodespan(\treeroot(T))$ for any node $u$.
Consequently, storing any $\start(u)$ takes only $\log n + (\log{n}) / b$ bits.
Storing $\bound(u)$ requires $\Oh(\log(b\cdot \log n))$ bits.
This results in an ancestry labeling scheme with labels of length $\log{n}+\Oh(\log{\log{n}})$,
for example by taking $b=\log{n}$.

Consult Figure~\ref{fig:example} for an example with assigned $\start$ and $\bound$ values.

\begin{figure}
  \centering
  \includegraphics[scale=0.8]{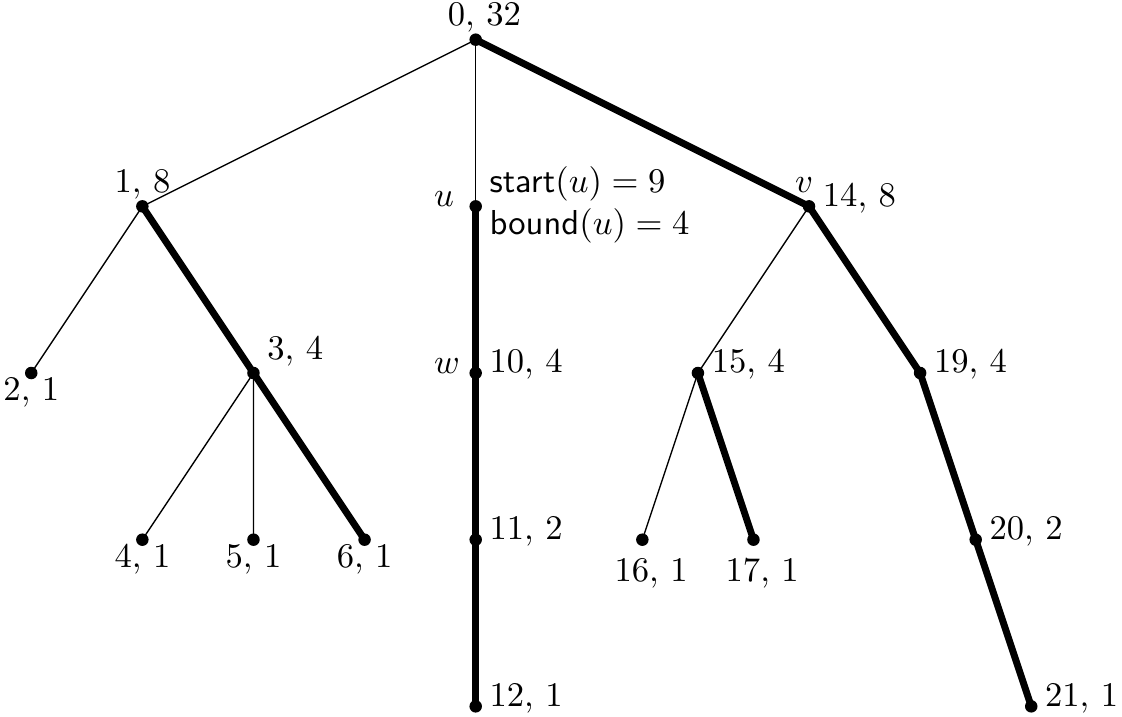}
  \caption{$\start$ and $\bound$ values for a given tree, with $b=1$ for simplicity.
  Thicker lines represent heavy paths.
  Note that $\start(u)+\bound(u)=13$ is smaller than $\start(w)+\bound(w)=14$, even though $w$ is a (heavy) child of $u$.}
  \label{fig:example}
\end{figure}

\subsection{Routing in trees of bounded degree}

The ancestry labeling scheme from the previous subsection extends easily to a routing labeling scheme for
bounded degree trees with labels of length $\log{n}+\Oh(\log{\log{n}})$. As explained in Section~\ref{sec:lower_bound},
this is optimal up to the asymptotics of the smaller-order term, even for very simple trees.

\begin{theorem}
There exists a canonical labeling scheme for routing in bounded degree trees on $n$ nodes
with labels of length $\log{n}+\Oh(\log{\log{n}})$ bits.
\end{theorem}

\begin{proofs}
We apply the ancestry scheme described in the previous section.
Observe that for a head of a heavy path $u_{1}$ we can make sure that $\nodespan(u_1) = \lfloor 2^{t/b}\rfloor$
for some $t$ by increasing $\bound(u_1)$ if necessary.
Such an operation introduces rounding by at most $2^{1/b}$ for every light edge.
With this property, every node $u_i$ can save in its label a \emph{routing table} $\rt(u_i)$,
storing a constant number of (rounded) numbers $\rt(u_i)[j]=\nodespan(v_{i,j})$
that represent rounded sizes of subtrees rooted in the light children of $u_i$.
Additionally, we assign number $j+1$ to the port leading to $v_{i,j}$, and number $1$ to the port leading to the heavy child.
Then, the decoder can answer a query for routing from $u$ to $w$ in the following way.
Firstly, it checks whether $u$ is an ancestor of $w$, and if not port number $0$ leading to the parent of $u$ is chosen.
Secondly, if $\start(w)-\start(u) > \sum_{j=1} \rt(u)[j]$, port number $1$ leading to the heavy child is chosen.
In the last case, the decoder finds the smallest $k$ such that $\start(w)-\start(u) \leq \sum_{j=1}^{k} \rt(u)[j]$ holds
and then port number $k+1$ is chosen.
Pseudocode for the encoder is presented in Algorithm~\ref{alg:prelimenc}.

Again, by a straightforward induction on the light depth of a node, it can be shown that $\nodespan(u) \leq |T_{u}| 2^{2\level(u)/b}$.
Using $b=\log{n}$, we get that the final length of a label is as stated, since $\start(u)$ is stored on $\log{n}+\Oh(1)$
bits and there is constant number of other parts, each stored on $\Oh(\log{\log{n}})$ bits.
The scheme can be made canonical by simply ordering light children by the size of their subtrees.
\end{proofs}

\begin{algorithm}[h]
\begin{algorithmic}[1]
  \Function{Create-labels}{$P$}
  \State \textbf{Input:} heavy path $P = u_1,\ldots,u_p$. $b$ is a fixed parameter.
  \State \textbf{Output:} labels for every node $u \in T_{u_1}$.
  \State
  \State $A \gets 0$ \Comment{Accumulator}
  \For{$i=1 \ldots p$}
    \State $\start(u_i) \gets A, A \gets A + 1$
    \For{$j=1 \ldots \degree(u_i)-1$} \Comment{Light children of $u_i$}
      \State \Call{Create-labels}{$P'$}, where $P'$ is a heavy path with $v_{i,j}$ as the head     
      \State Add $A$ to $\start(w)$ for every $w \in T_{v_{i,j}}$ \Comment{Shifting procedure}
      \State $\rt(u)[j] \gets \nodespan(v_{i,j})$ \Comment{Routing table}
      \State $A \gets A + \nodespan(v_{i,j})$
    \EndFor
  \EndFor
  \For{$i=1 \ldots p$}
  	\State Let $t$ be the smallest natural number such that $\lfloor 2^{t/b} \rfloor + \start(u_i) \geq A$
    \State $\bound(u_i) \gets \lfloor 2^{t/b} \rfloor$
  \EndFor
  \State Let $t$ be the smallest natural number such that $\lfloor 2^{t/b} \rfloor \geq \nodespan(u_1)$
  \State $\bound(u_1) \gets \lfloor 2^{t/b} \rfloor$
  \EndFunction
\end{algorithmic}
\caption{The encoder for trees of bounded degree.}
\label{alg:prelimenc}
\end{algorithm}

\begin{algorithm}[h]
\begin{algorithmic}[1]
  \Function{Get-port}{$\ell(u),\ell(w)$}
  \State \textbf{Input:} labels of $u$ and $w$.
  \State \textbf{Output:} port number corresponding to the first edge on the path from $u$ to $w$.
  \State
  \State Unpack $\start(w), \start(u), \bound(u), \rt(u)$ from the labels
  \If{$\start(w) \not \in (\start(u),\start(u)+\bound(u))$}
  	\State \Return{0} \Comment{Not an ancestor, routing up}
  \EndIf
  \State $S \gets 0$
  \For{$j=1 \ldots \deg(u)-1$} \Comment{Iterating through the routing table}
  	\State $S \gets S + \rt(u)[j]$
  	\If{$\start(w)-\start(u) \leq S$} \Return{$j+1$}
  	\EndIf
  \EndFor
  \State \Return{1} \Comment{Routing to the heavy child}
  \EndFunction
\end{algorithmic}
\caption{The decoder for trees of bounded degree.}
\label{alg:prelimdec}
\end{algorithm}

In an example presented in Figure~\ref{fig:example}, we have $\nodespan(u)=5$,
so (for $b=1$) it would be rounded to $\nodespan(u)=8$ just by setting $\bound(u)=8$.
Then the accumulator would carry that shift further, and for example $\start(v)$ would be set to $17$.

\subsection{Simpler preliminary routing scheme}
The goal of this subsection will be to prove the following theorem:

\begin{theorem}
There exists a canonical labeling scheme for routing in trees on $n$ nodes with labels of length $\log n + \Oh((\log^{3/4}n)\sqrt{\log{\log{n}}})$.
\end{theorem}

Although already better than previously known results, it can be seen as an extended warm-up
and methods used in this subsection are not essential for the rest of the paper.
Therefore the reader comfortable enough with the framework can safely skip this part.

We start with a general overview of a routing labeling scheme.
Given the value of $\start$, we need to be able to determine the child whose subtree contains the corresponding node.
With this goal in mind, we partition the light children of every node into small and big.
To define the partition, we say the light child $v$ of $u$ is \emph{big} when $\level(v)\in (\level(u)-c,\level(u))$,
where $c$ is a parameter to be chosen later.
Otherwise light child $v$ is \emph{small}.
We arrange the light children of $u$ in non-increasing order by the sizes of their subtrees,
so that we first have all big children and then all small children.
The intervals of big children are processed exactly as in an ancestry scheme,
but for the intervals of all small children we introduce an intermediate step.

Observe that on any path from the root to a leaf there are at most $(\log n) / c$ small children, 
because every small child implies that the level decreases by at least $c$.
This allows us to deal with small children via a costly but simple approach.
For a node $u_{i}$ with small children $v_{i,k+1},v_{i,k+2},\ldots,v_{i,k+s}$
let $\sma(u_{i}) = \sum_{j=1}^{s} \nodespan(v_{i,k+j})$.
As this value will be stored in a label, we need to round it up to $\smap(u_i)$,
the smallest number $x = \lfloor 2^{t/b} \rfloor$ such that $x \geq \sma(u_i)$.
Then we conceptually replace all intervals corresponding to the small children
of $u_{i}$ with a single interval of length $\smap(u_{i})\log{\smap(u_{i})}$.
This interval is partitioned according to a harmonic sequence, so at the beginning there is
a subinterval of size $\smap(u_{i})$, then subinterval $\lfloor \smap(u_{i})/2 \rfloor$,
in general the $j$-th subinterval has size $\lfloor \smap(u_{i})/j \rfloor$ for $j=1,2,\ldots,\smap(u_{i})$.
Clearly, due to the ordering of the children, $\nodespan(v_{i,k+j})$ can fit into the $j$-th subinterval.
Moreover, knowing $\smap(u_i)$ the decoder can compute the length $\smap(u_{i})\log{\smap(u_{i})}$
of the created interval. Then, given any number from such an interval, the decoder can determine which 
subinterval contains this number, so also which child contains the corresponding node in its subtree.
The final size of the single interval reserved for storing all small children is at most $\sma(u_{i})2^{1/b}(\log{\sma(u_{i})}+1)$.

After processing the small children, we apply the shifting procedure to all intervals corresponding to the big children of $u_{i}$
and the single interval corresponding to its small children.
The shift of the single interval corresponding to all small children of $u_{i}$
is then used to define the shifts of all intervals corresponding to these children, according to a harmonic sequence.
Consult Figure~\ref{fig:line} for an illustration of these intervals on the number axis.
Finally, we make sure that $\nodespan(u_{1})=\lfloor 2^{t/b} \rfloor $ for some $t$ by increasing $\bound(u_{1})$ if necessary.
This is to guarantee that there is a small number of the possible values of $\nodespan$ for the heads of heavy paths, which will be useful later.
By the same argument as before, the procedure guarantees the properties stated at the beginning of Section~\ref{sec:preliminary},
but now we need to analyze the resulting $\nodespan(\treeroot(T))$ more carefully.

\begin{figure}
  \includegraphics[width=\textwidth]{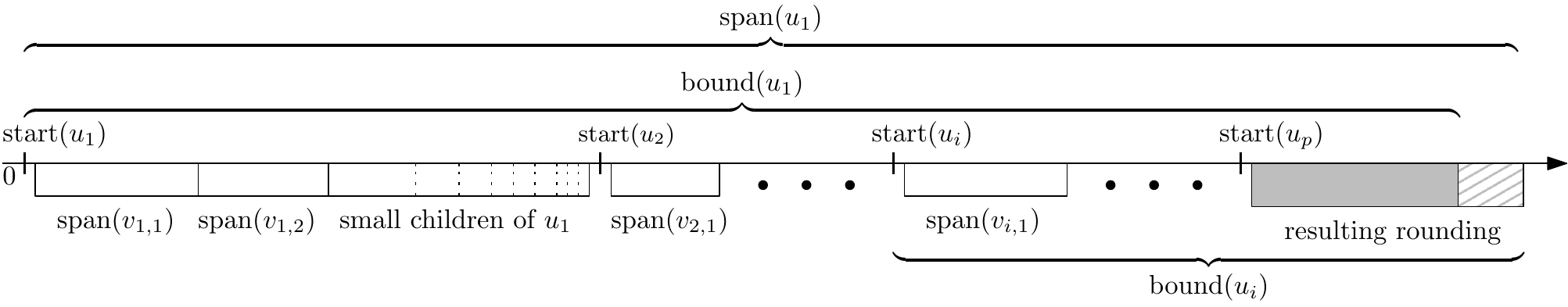}
  \caption{Intervals corresponding to the nodes of a single heavy path and the subtrees that hang off this path.
  $\nodespan(u)$ might be larger than $\bound(u)$, although for the head of the heavy path $u_1$
  we will make sure that $\nodespan(u_1)=2^{t/b}$.}
  \label{fig:line}
\end{figure}

\begin{claim}
Let $u_1$ be the head of its heavy path.
For $n$ large enough the shifting procedure ensures that, after having processed $u_1$, the following holds:
\begin{align*}
\nodespan(u_1) \leq |T_{u_1}| 2^{2\level(u_1)/b} (2\log n)^{\lfloor \level(u_1)/c \rfloor}
\end{align*}
\label{prop:size0}
\end{claim}

\begin{proof}
Observe that the above expression can be upper bounded as follows:
\begin{align*}\nodespan(u_1) &\leq |T_{u_1}| 2^{2\level(u_1)/b} (2\log n)^{\lfloor \level(u_1)/c \rfloor} \\
& \leq n 2^{2\log n / b} (2\log n)^{\log n / c} 
 = n 2^{2\log n / b + (1 + \log\log n) \log n / c}. \end{align*} 
As long as $b=\omega(1)$ and $c=\omega(\log\log n)$, for $n$ large enough this is at most $n^{1.5}$.
We will use this crude bound on $\nodespan(u)$ to simplify some of the calculations.

Claim~\ref{prop:size0} is proved by induction on the light depth of a node.
The induction basis is a heavy path with no light children hanging off it, for which we have
just $\nodespan(u_1) = 1$ if $|T_{u_1}|=1$ and $\nodespan(u_1) \leq |T_{u_1}|2^{1/b}$ otherwise.
Now let us consider a heavy path $P=u_{1}-u_{2}-\ldots -u_{p}$
and assume that the inequality holds for the root $v_{i,j}$ of every subtree hanging off the path:
\[ \nodespan(v_{i,j}) \leq |T_{v_{i,j}}| 2^{2\level(v_{i,j})/b} (2\log n)^{\lfloor \level(v_{i,j})/c \rfloor} \]
Then, replacing the intervals of all small children of $u_{i}$ with a single interval
increases their total length by a factor of $2^{1/b}(\log{\sma(u_{i})}+1)$, where for $n$ large enough
$\sma(u_{i}) \leq n^{1.5}$ and $\log{n^{1.5}+1} \leq 2\log{n}$, so the increase is by at most a factor of $2^{1/b}2\log{n}$.
If $v_{i,j}$ is a small child, then $\lfloor \level(v_{i,j})/c\rfloor \leq \lfloor\level(u_{1})/c\rfloor -1$, so we have
\begin{linenomath}\begin{align*}
\smap(u_{i})\log{\smap(u_{i})} &\leq 2^{1/b}2\log{n} \sum_{j: v_{i,j} \textrm{ is small}} \nodespan(v_{i,j}) \\
& \leq (2\log n)^{\lfloor \level(u_{1})/c \rfloor} \sum_{j: v_{i,j} \textrm{ is small}} |T_{v_{i,j}}| 2^{(2\level(v_{i,j})+1)/b} \\
& \leq 2^{(2\level(u_1)-1)/b}(2\log n)^{\lfloor \level(u_{1})/c \rfloor} \sum_{j: v_{i,j} \textrm{ is small}} |T_{v_{i,j}}|
\end{align*}\end{linenomath}
For any big child $v_{i,j}$, we have $\nodespan(v_{i,j}) \leq 2^{(2\level(u_1)-1)/b} (2\log n)^{\lfloor \level(u_{1})/c \rfloor}|T_{v_{i,j}}|$.
After that, $\bound$ rounding at the end of the shifting procedure adds another factor of $2^{1/b}$, and at this moment:
\begin{linenomath}\begin{align*}
\nodespan(u_{1}) &\leq 2^{(2\level(u_{1})-1) / b} (2\log n)^{\lfloor \level(u_{1})/c\rfloor}\sum_{i}^{p} (1+ \sum_{j}|T_{v_{i,j}}|) \\
& = |T_{u_{1}}| 2^{(2\level(u_{1})-1) / b} (2\log n)^{\lfloor \level(u_{1})/c\rfloor}
\end{align*}\end{linenomath}
But we might increase $\bound(u_{1})$ by another factor of $2^{1/b}$, so indeed finally we have:
\[\nodespan(u_{1}) \leq |T_{u_{1}}| 2^{2\level(u_{1})/b} (2\log n)^{\lfloor \level(u_{1}) /c\rfloor} \qedhere \] 
\end{proof}

Thus:
\begin{linenomath}\begin{align*}
\nodespan(\treeroot(T)) &\leq n 2^{2\log n /b} (2\log n)^{\log n/c} 
 \leq 2^{\log n + \Oh(\log n / b + \log \log n \cdot \log n / c) }
\end{align*}\end{linenomath}
and storing the final $\start(u)$ for any $u\in T$ takes $\log n + \Oh(\log n / b + \log\log n \cdot \log n / c)$ bits,
while storing $\bound(u)$ takes only $\Oh(\log (b\cdot\log n))$ bits, under earlier assumptions on $b$ and $c$.
Furthermore, now by inspecting the proof we obtain that, for every head of a heavy path $u_1$,
$\bound(u_1)$ value can be always modified so that
$\nodespan(u_1) / ( |T_{u_1}| 2^{2\level(u_1) /b} (2\log n)^{\lfloor \level(u_1)/c \rfloor} ) \in (2^{-1/b},1]$.

\paragraph{Encoder.}
We are now ready to present the details of our routing scheme.
Ports are numbered just according to the non-increasing size of subtrees, with port 1 leading to the heavy child.
$\start(u)$ and $\bound(u)$ are stored inside $\ell(u)$.
Additionally, every node $u\in T$ stores information about its big children that we call its routing table.
For any big child $v_{j}$, we have that $\level(v_{j})\in (\level(u)-c,\level(u))$.
We also made sure that $\nodespan(v_{j})=2^{t/b}$ for some $t$, as $v_j$ is a head of its heavy path.
Finally, $\nodespan(v_{j}) / ( |T_{v_{j}}| 2^{2\level(v_{j})/b} (2\log n)^{\lfloor \level(v_{j})/c \rfloor} ) \in (2^{-1/b},1]$.
Thus, the number of possible distinct values of $\nodespan(v_{j})$
is not larger than $b\cdot \log(2^{c+1}2^{2c/b}\log{n})$, which is $\Oh(b \cdot c)$ assuming $c = \Omega(\log{\log{n}})$.
For each of these $\Oh(b\cdot c)$ distinct values, we store how many big children $v_{j}$ of $u$
have such $\nodespan(v_{j})$. This number is at most $2^{c}$, so all these values take
$\Oh(b\cdot c^{2})$ bits in total.

\paragraph{Final structure of a label.}
Label $\nodelabel(u)$ of node $u$ is composed of several parts:
\begin{enumerate}
\item $\start(u)$,
\item $\bound(u)=\lfloor 2^{t/b}\rfloor$ stored as $t$ in binary,
\item $\smap(u)$,
\item routing table $\rt(u)$ storing the number of big children of every possible size,
\item $\level(u)$.
\end{enumerate}
\noindent
Constants $b$ and $c$ will be computable from the value of $\lceil \log{n} \rceil$.
All labels are distinct because $\start$ values are distinct.

\paragraph{Decoder.}
As for decoding, with $\nodelabel(u)$ and $\nodelabel(w)$ we first check whether we should choose port 0
to the parent of $u$ (when $\start(w) \not \in [\start(u),\start(u)+\bound(u))$).
In the other case, $u$ is an ancestor of $w$, and we need to route down.
Entries of the routing table storing the number of big children of given size are checked one by one,
while maintaining the sum of the values of $\nodespan$ for the already considered children,
until this sum is at least $\start(w)-\start(u)$, indicating the correct port number.
If this does not happen, known value of $\smap(u)\log{\smap(u)}$ allows us to check whether $w$ is in the subtree of some small child.
If so, the correct port can be computed from a harmonic sequence.
The last remaining option is port $1$ leading to the heavy child, which is chosen when $\start(w)-\start(u)$ is greater than the
sum of spans of all big children and the interval of light children.

\paragraph{Length of a label.}
Now we can analyse the total length of a label:
\begin{align*}
 & \log n + \Oh(\log n / b + \log\log n \cdot \log n / c)
+ \Oh(\log (b\cdot\log n)) + \Oh(b\cdot c^{2}) \\
&=  \log n + \Oh(\log n / b + b\cdot c^{2} + \log\log n \cdot \log n / c)
\end{align*}
We minimise the above expression by setting $b=(\log^{1/4}n)/\sqrt{\log{\log{n}}}$ and $c=(\log^{1/4}n)\sqrt{\log{\log{n}}}$.
This results in labels of total length $\log n + \Oh((\log^{3/4}n)\sqrt{\log{\log{n}}})$.
The obtained labeling scheme is canonical, as the encoder assigns port 1 to the heavy child even though it is
processed after all the other children.

\section{Intermediate scheme with double rounding}
\label{sec:improved2}

In this section we will prove the following theorem:

\begin{theorem}
There exists a canonical labeling scheme for routing in trees on $n$ nodes with labels of length
$\log{n} + \Oh(\sqrt{\log{n}\log{\log{n}}})$.
\end{theorem}

The method will be further refined in the next section to arrive at our final result.
From now on we slightly change the high-level way in which assignment of labels works,
splitting it into two phases for the encoder.
The main first phase works in a bottom-up manner and provides the necessary statistics for creating labels.
Every node in a tree is assigned the size of its \emph{reserved segment}, and the \emph{routing tables} are created.
Then the second top-down phase deals with assigning proper values of $\start(u)$ and $\bound(u)$,
with the sizes of the reserved segments guaranteeing such assignment to be possible.
For this section, we could do with just one phase, but this distinction will be useful later.

We use a positive integer parameter $b$ to be fixed later.
The main concept is that, like before, we assign to nodes intervals of numbers from range
$[1,n2^{\Oh(\log{n}/b)}]$, so that the IDs can be stored on $\log{n}+\Oh(\log{n}/b)$ bits,
and we allow $\nodespan$ to be rounded up by a factor of $2^{\Oh(1)/b}$ for every light edge.
We will use about $b\log{\log{n}}$ bits in each node to store a routing table.
Moreover, we denote by $\seglen(u_1)$ the length of the reserved segment for $u_1$ being the head of its heavy path.
Our intention is that $\seglen(u_1)$ should be guaranteed to be at least as big as $\nodespan(u_1)$,
so given any interval of length at least $\seglen(u_1)$ we should be able to assign unique 
$\start$ and proper $\bound$ values in the whole subtree $T_{u_1}$.
To facilitate efficient encoding, $\seglen(u_1)$ will be rounded up two times when processing the parent of $u_1$.
Firstly, there is rounding to $\seglenp(u_1)$, when we ensure there is only a small set of possible values for sizes of the children's segments.
Secondly, rounding to $\seglenb(u_1)$ happens, when we gather children in groups and make the segment size of everyone in a group equal.
$\rt(u)$ will denote a binary string representing the routing table for $u$ used to route down to the children of $u$.

\subsection{Encoder}
\paragraph{First phase --- segments assignment.}
First phase of the encoder generates rounded segment sizes and routing tables.
The respective values are assigned in a bottom-up manner, guided by the heavy path decomposition.
Consider a heavy path $P=u_{1}-u_{2}-\ldots-u_{p}$,
where $\head(P)=u_{1}$, and assume that we have already visited all nodes in the subtrees hanging off $P$.
By $\lightweight(u)$ we denote the \emph{light weight} of a node $u$,
that is the sum of subtrees' sizes of its light children.
To make calculations less technical and reduce the number of cases, in the following we will 
assume that for every node $u$ $\lightweight(u) \geq c'$, for some constant $c'$.
This is obviously not true for the input tree, but as explained later
it can be achieved for example by adding $c'$ leaf children to every node of the input tree
and making some simple adjustments to the encoder to handle these leaves.
For now, we assume that $\lightweight(u) \geq 4$, so that $\lfloor \log{\log{\lightweight(u)}} \rfloor$ is non-zero.
With a scheme from this section we are able to achieve the following:

\begin{claim}
Under assumption on $\lightweight$ values made above and for $b \geq 6$, using $\seglen(u_{1}) = |T_{u_{1}}|2^{12\level(u_{1})/b}$
for every $u_1$ being the head of a heavy path is sufficient for the encoder to be able to store $\rt(u)$
on only $\Oh(b\log{\log{\lightweight(u)}})$ bits for every node $u$.
\label{prop:size1}
\end{claim}

This will be proved inductively for the procedure described below.
The procedure traverses the path from its head, one node at a time.
Let $v_{i,1},v_{i,2},\ldots,v_{i,\degree(u_i)-1}$ be the light children of $u_{i}$,
sorted in the non-increasing order by the values of $\seglen(v_{i,j})$
(note that this is also non-increasing by the values of $|T_{v_{i,j}}|$).
In the following, we will often refer to light children as just children, since the heavy child $u_{i+1}$ will
be handled separately later and it does not introduce rounding.
Important observation is that often we can afford more rounding than just by a factor of $2^{\Oh(1)/b}$.
In fact, if $\level(v_{i,j}) = \level(u_{i}) - k$, then rounding by $2^{\Oh(k)/b}$
still guarantees a good final bound, as we allow rounding by $2^{\Oh(1)/b}$ for every skipped level.
We will now use this observation to achieve small number of possible children sizes.

\paragraph{Rounding in classes.}
Let $l = \min(\lfloor \log{\lightweight(u_{i})} \rfloor + 1, \level(u_i))$,
so that $\level$ of any light child of $u_i$ is less than $l$. 
We define preclasses as (possibly empty) sets of children of $u_{i}$ with the same value of $\level$,
with children from the first preclass having $\level$ equal to $l-1$, from the second preclass $\level$ $l - 2$ and so on.
Then classes are defined as follows: for the first $b-1$ preclasses, the $k$-th preclass is evenly divided into $\lceil b/k \rceil$ classes.
So if the $k$-th preclass consists of children with $\level$ equal to $x$, then the sizes of their subtrees are in $[2^{x},2^{x+1})$,
and the first class created by dividing this preclass consists of children with the size of the subtree in $[2^{x},2^{x+k/b})$,
the second class in $[2^{x+k/b},2^{x+2k/b})$, and so on, the last one possibly having smaller interval.
This process of subdividing the first $b-1$ preclasses creates at most $\sum_{j=1}^{b-1} \lceil b/i \rceil \leq b(\ln{b}+1) + b \leq b(\log{b}+2)$ classes.
Preclasses with rank at least $b$ are not divided but merged into classes.
Preclasses with ranks from $b$ to $2b-1$ are just left as $b$ separate classes, then preclasses 
$2b, \ldots, 4b-1$ are merged in pairs into $b$ classes, the next $4b$ preclasses are merged in quadruples to obtain another $b$ classes, and so on.
The last class might be composed of a number of preclasses not being power of two.
It can be seen that at most $b \lceil \log{(\log{\lightweight(u_{i})}/b)} \rceil$ classes are created.
In total we have no more than  $b(\log{b} + \log{(\log{\lightweight(u_{i})}/b)} + 3) \leq 
b(\log{\log{\lightweight(u_{i})}}+3)$ classes.
Pseudocode for this procedure is presented in Algorithm~\ref{alg:classes}. 

\begin{algorithm}[h]
\begin{algorithmic}[1]
  \Function{Construct-classes}{$u_i$, $b$}
  \State \textbf{Input:} node $u_i$ and its light children $v_{i,1},\ldots,v_{i,\degree(u_i)-1}$, parameter $b$.
  \State \textbf{Output:} partition of light children into classes.
  \State
  \State $PC \gets \emptyset$ \Comment{Division into preclasses}
  \State $l \gets \min(\lfloor \log{\lightweight(u_{i})} \rfloor + 1, \level(u_i))$
  
  \For{$k=1 \ldots l$}
    \State $pc_k \gets \{v_{i,j} : \level(v_{i,j}) = l-k\}$ \Comment{Preclass is the set of children with the same $\level$}
    \State $PC \gets PC \cup \{pc_k\}$
  \EndFor
   
  \State
  \State $C \gets \emptyset$ \Comment{Division into classes}
  
  \For{$k=1 \ldots b-1$} \Comment{First $b-1$ preclasses are subdivided}
  	\For{$p=1 \ldots b/k$}
  		\State $cl \gets \{v_{i,j} : |T_{v_{i,j}}| \in [2^{l-k+(p-1)k/b},
  		2^{l-k+pk/b})\}$
  		\State Set $2^{l-k+pk/b}2^{12(l-k)/b}$ as the boundary value for class $cl$
  		\State $C \gets C \cup \{cl\}$
  	\EndFor
  \EndFor
  
  \State $j \gets b, z \gets 1$
  \State $cap \gets l-b+1$
  \While{$j \leq l$} \Comment{The remaining preclasses are merged into classes}
  	\For{$k=1 \ldots b$}
  	\State $cl \gets \bigcup\limits_{m=j+(k-1)z}^{\text{min}(j+kz-1,l)} pc_m$
  	\State Set $2^{cap(1+12/b)}$ as the boundary value for class $cl$
  	\State $C \gets C \cup \{cl\}$
  	\State $cap \gets cap-z$
  	\EndFor
  	\State $j \gets 2j, z \gets 2z$
  \EndWhile
  \State \Return{C}
  \EndFunction
\end{algorithmic}
\caption{Dividing light children of a node into classes according to their sizes.}
\label{alg:classes}
\end{algorithm}

A \emph{boundary value} of a class is an upper bound on the possible $\seglen$ value in this class.
If a class consists of every light child $v_{i,j}$ with $|T_{v_{i,j}}| \in [x_1,x_2+1)$,
then the boundary value for this class is $x_2 2^{12 \lfloor \log{x_2} \rfloor /b}$,
as $\level(u) = \lfloor \log{|T_u|} \rfloor$ and $\seglen(u) = |T_{u}|2^{12\level(u)/b}$.
With division into classes as described, we achieve the promised property
that rounding factor depends heavily on the level of a child:

\begin{lemma}
Assuming $b \geq 6$, if a node $v_{i,j}$ with $\level(v_{i,j}) = l - k$ is part of class $C$,
then the boundary value for this class is no larger than $\seglen(v_{i,j})2^{3k/b}$.
\label{lem:roundclasses}
\end{lemma}
\begin{proof}
If $C$ is one of the classes subdividing some preclass of rank at most $b-1$,
then we have that the boundary value for this class is actually no larger than $\seglen(v_{i,j})2^{k/b}$,
as interval of $C$ spans inside just a single level subdivided into $b/k$ classes.
Now assume that $C$ is a class constructed by merging $r$ preclasses of rank between $rb$ and $2rb-1$.
Then $\level$ of two children in $C$ must differ by less than $r$, which means
the sizes of their subtrees differ by less than a factor of $2^r$.
Thus, the boundary value for $C$ is less than $\seglen(v_{i,j})2^r2^{12r/b}$.
Since $b \geq 6$, this is at most $\seglen(v_{i,j})2^{3r}$.
As $C$ is constructed by merging preclasses of rank at least $rb$,
then $\level(v_{i,j}) \leq l - rb$, so we have that $k \geq rb$ and
finally the boundary value for $C$ is at most $\seglen(v_{i,j})2^{3k/b}$.
\end{proof}

For every class $C$, the encoder rounds the segment length of every child in $C$ up to the boundary value for $C$.
This way, rounding in classes will increase the size of a reserved segment of $v_{i,j}$ to $\seglenp(v_{i,j})$, with
\[ \seglenp(v_{i,j}) \leq \seglen(v_{i,j})2^{3(l - \level(v_{i,j}))/b} \leq \seglen(v_{i,j})2^{3(\level(u_i) - \level(v_{i,j}))/b} \]
and the last inequality being sufficient for this section.
By rounding in classes, we achieved a small number of different possible segment lengths without too much rounding.

Note that the set of possible values of $\seglenp(v_{i,j})$ (boundary values for classes) 
depends only on $b$ and $l$, not the actual children sizes.
This is crucial, as it allows for compact encoding of the routing table if we know these two values.
Still, there can be as many as $\Omega(n)$ children with a given rounded segment length,
which prohibits us from just storing their cardinality directly in $\rt(u_{i})$.
This problem can be once again solved by grouping and rounding up.

\paragraph{Rounding in groups.}
After rounding in classes, we have at most $z = b(\log{\log{\lightweight(u_{i})}}+3)$ different possible sizes of segments ($\seglenp$ values).
Our second goal is to divide children into at most $z$ groups of fixed sizes and make their segment lengths equal.
Then we will be able to use the following:

\begin{proposition}
\label{fact:plot}
A non-increasing sequence of at most $z$ integers from $[0 , z]$ can be encoded on $2z$ bits.
\end{proposition}
\begin{proof}
The decoder starts with a counter set to $z$.
Then a sequence is stored as a bit string in which 0 tells the decoder to decrease the counter by one,
and 1 tells the decoder that the next element of the sequence is equal to the current counter.
\end{proof}

Elements of the sequence will correspond to consecutive groups of children, and their values to
the set of boundary values defined by the classes.
Thus, these values correspond to the rounded sizes of reserved segments in groups.
To achieve the stated goal, we define pregroups and groups, similarly to preclasses and classes.
Children of $u_{i}$ are divided into at most $l$ pregroups,
the $k$-th one containing $2^k$ consecutive children, sorted in the non-increasing order by the subtree size, or equivalently $\seglenp$.
The last pregroup is padded with dummy children of size 0 if needed, as we cannot afford storing the exact degree of a node.
Then, as before, groups are defined in the following way:
for the $b-1$ first pregroups, the $k$-th pregroup is divided into $\lceil b/k \rceil$ groups with roughly equal number of children
(some are possibly empty).
Pregroups with rank at least $b$ are not divided but merged into groups.
Pregroups with ranks from $b$ to $2b-1$ are just left as $b$ separate groups, then pregroups 
$2b \ldots 4b-1$ are merged in pairs into $b$ groups, the next $4b$ pregroups are merged in quadruples
for another $b$ groups, and so on.
Total number of created groups is not greater than the number of classes.
The encoder, for every group, rounds the reserved segment length of every child in this group up to the length of the largest one.
Pseudocode for rounding in groups is presented in Algorithm~\ref{alg:groups},
and Figure~\ref{fig:plot1} depicts rounding in both classes and groups.

\begin{algorithm}
\begin{algorithmic}[1]
  \Function{Construct-groups}{$u_i$, $b$}
  \State \textbf{Input:} node $u_i$ and its light children $v_{i,1},\ldots,v_{i,\degree(u_i)-1}$, parameter $b$.
  \State \textbf{Requirement:} light children in the non-increasing order by their size.
  \State \textbf{Output:} partition of light children into groups.
  \State
  \State $PG \gets \emptyset$ \Comment{Division into pregroups}
  \State $j \gets 1, c \gets 0$
  \While{$j < \degree(u_i)$}
    \State $c \gets c+1$
    \State $pg_{c} \gets \{v_{i,j},\ldots,v_{i,\text{min}(\degree(u_i)-1,2j)}\}$
    \State $PG \gets PG \cup \{pg_c\}$
  	\State $j \gets 2j+1$
  \EndWhile
  \If{$|pg_c|$ is not power of two}
  	\State Add to $pg_c$ sufficiently many dummy children of size 0
  \EndIf
  \State
  
  \State $G \gets \emptyset$ \Comment{Division into groups}
  
  \For{$j=1 \ldots b-1$} \Comment{First $b-1$ pregroups are subdivided}
    \State Order the children in $pg_{j}$ by size
  	\State Divide $pg_j$ into $b/j$ consecutive groups of roughly equal number of children and add them to $G$
  \EndFor
  \State $j \gets b, z \gets 1$
  \While{$j \leq c$} \Comment{The remaining pregroups are merged into groups}
  	\For{$k=1 \ldots b$}
  	\State $g \gets \bigcup\limits_{m=j+(k-1)z}^{\min(j+kz-1,c)} pg_m$
  	\State $G \gets G \cup \{g\}$
  	\EndFor
  	\State $j \gets 2j, z \gets 2z$
  \EndWhile
  \State \Return{G}
  \EndFunction
\end{algorithmic}
\caption{Dividing light children of a node into groups in which the segment lengths are rounded up to the same value.}
\label{alg:groups}
\end{algorithm}

Note that the number of children in any group depends only on $b$, and the number of pregroups
depends only on $\degree(u_i)$, which is crucial for encoding of the routing table.
Only the number of dummy children cannot be reproduced from these values.
Moreover, double rounding still preserves the monotonicity of sizes of the children subtrees, so our routing scheme will be canonical.
Now we need to prove a bound on the total increase in lengths of segments.

\begin{figure}
\begin{center}
  \includegraphics[scale=0.75]{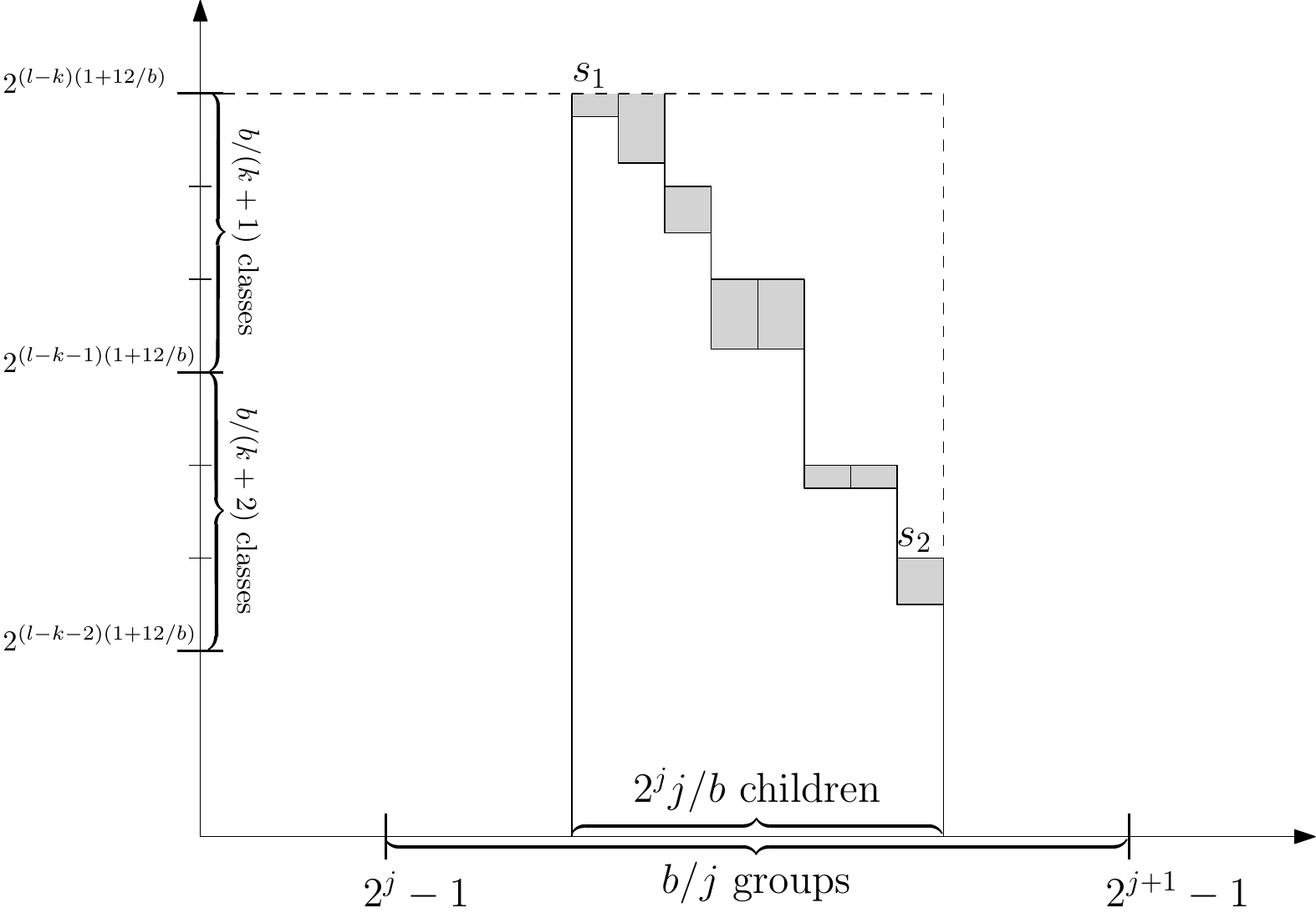}
\end{center}
  \caption{Rounding of one group, with children on the horizontal axis and their segment lengths on the vertical axis, $k < b$.
  Children are rounded up to the boundary value for their class,
  and then the whole group is rounded up to maximum $\seglenp$ value in this group, $s_1$.}
  \label{fig:plot1}
\end{figure}

\begin{lemma}
Let $\seglenp(v_{i,j})$ be the size of the reserved segment of $v_{i,j}$ after rounding in classes
and $l = \min(\lfloor \log{\lightweight(u_{i})} \rfloor + 1, \level(u_i))$.
Then the total size of segments after rounding in groups for light children of $u_{i}$, including dummy ones, is bounded by
$\sum_{j\geq 1} \seglenp(v_{i,j})2^{8(l-\level(v_{i,j}))/b}$,
provided that $b \geq 6$.
\label{lem:potentials}
\end{lemma}

\begin{proof}
Observe that for child $v_{i,j}$ in the $k$-th pregroup we have $\level(v_{i,j}) \leq l - k$.
We assign potentials to children, child $v_{i,j}$ in the $k$-th pregroup is assigned value
$\seglenp(v_{i,j})(2^{8k/b}-1)$, so if we think about the plot (as in Figure~\ref{fig:plot1})
it corresponds to $(2^{8k/b}-1)$ units of potential assigned to every unit length of the column
representing $\seglenp$ value of a child from the $k$-th pregroup.
Then we claim that the sum of children's sizes increases during the rounding is bounded by the sum of all assigned potentials.

To see this, first consider the case of a group inside the $k$-th pregroup, for $k < b$.
It consists of at most $2^kk/b$ children.
Let $s_1$ be the size of the first (the largest) of them, and $s_2$ the last (the smallest) one.
Then the size of every child is rounded up to $s_1$.
Note that in previous groups the size of every child was at least $s_1$ even before rounding in these groups,
and that children from further groups will not be rounded to more than $s_2$.
Thus for rounding in the considered group we can charge unit lengths from $s_1$ to $s_2$ for children in all prior groups.
Charging only the $(k-1)$-th pregroup will be enough for us.
In the first pregroup there is no rounding at all, provided that $b>1$.
As $1 + x \leq e^{x}$ and thus $2^{2x} - 1 \geq x$,
sum of the potentials of unit lengths from $s_1$ to $s_2$ for elements from the $(k-1)$-th pregroup is equal to
\[ (s_1-s_2)2^{k-1}(2^{8(k-1)/b} -1) \geq (s_1-s_2)2^{k-1}4(k-1)/b \geq (s_1-s_2)2^kk/b \]

Each of the remaining groups is created by merging some pregroups.
Let $g$ be the number of groups created from the first $b-1$ pregroups (we already dealt with them).
For the increase of segments in a group of rank $g+r$, we can similarly charge only the previous group $g+r-1$,
except for the special case of $r=1$ where we would charge two groups $g$ and $g-1$ (which together are just the $(b-1)$-th pregroup).
Let $r \in (zb, zb+b]$ for $z \geq 0$, then the group of rank $g+r$ is created by merging
$2^{z}$ pregroups.
It can be seen that the children from group $g+r$ are from pregroups of rank at least $b2^z$,
and that the children from group $g+r-1$ are from pregroups of rank at least $b2^z-2^{z-1}-1$.
Thus, we will have a large reserve of potentials, as due to $b \geq 6$ every unit length for a child in group $g+r-1$ has potential at least
\begin{linenomath} \begin{align*}
& 2^{8(b2^z-2^{z-1}-1)/b}-1 = 2^{8 \cdot 2^z - 8(2^{z-1}+1)/b} -1 \\
 &\geq 2^{8 \cdot 2^z - 8(2^{z-1}+1)/6} -1
 \geq 2^{8 \cdot 2^z - 2^z - 2} -1 \\
 &\geq 2^{5 \cdot 2^z} -1 \geq 2^{2^z+1} 
\end{align*} \end{linenomath}
Since the number of children in group $g+r$ is at most $2^{2^{z}+1}$ times larger than the number 
of children in group $g+r-1$, an increase in the size of segments in group $g+k$ can be 
bounded by the sum of potentials of respective unit lengths from $s_1$ to $s_2$ from group $g+k-1$, as before.
\end{proof}

\begin{algorithm}[h]
\begin{algorithmic}[1]
  \Function{Assign-$\seglen$}{$P$}
  \State \textbf{Input:} heavy path $P = u_1,\ldots,u_p$. $b$ is a fixed parameter.
  \State \textbf{Requirement:} all nodes in the subtrees hanging off $P$ are already processed.
  \State \textbf{Output:} $\rt(u_i)$ and $\seglenb(v_{i,j})$ values for every $u_{i}$.
  \State
  \State $r \gets 0$
  \For{$i=1 \ldots p$}
    \State $C \gets$ \Call{Construct-classes}{$u_i$, $b$}
    \For{$j=1 \ldots \degree(u_i)-1$} \Comment{Light children of $u_i$}
      \State Set $\seglenp(v_{i,j})$ to be the boundary value of the class of $v_{i,j}$ in $C$
      \Comment{First rounding}
    \EndFor
    \State $G \gets$ \Call{Construct-groups}{$u_i$, $b$} \Comment{Note that dummy nodes are added here}
    \State $segment_i = 1$
    \For{$j=1 \ldots \degree(u_i)-1$}
      \State Set $\seglenb(v_{i,j})$ to be $\seglenp$ of the largest child in group of $v_{i,j}$ in $G$
      \Comment{Second rounding}
      \State $segment_i \gets segment_i + \seglenb(v_{i,j})$
    \EndFor
    \State $r \gets r + segment_i$    
    \State Set $\rt(u_i)$ as a description of rounded groups \Comment{Create the routing table}
  \EndFor
  \State \textbf{Guaranteed property:} $r \leq |T_{u_{1}}|2^{(12\level(u_{1})-1)/b}$
  \State $\seglen(u_1) \gets |T_{u_{1}}|2^{12\level(u_{1})/b}$
  \EndFunction
\end{algorithmic}
\caption{First phase of the encoder, assigning $\seglenb$ and $\rt$ on a heavy path.}
\label{alg:1phase}
\end{algorithm}

\paragraph{Bound on the length of a segment.}
After rounding in classes and groups for a node $u_i$,
$\rt({u_i})$ is set to be description of the rounded groups as discussed in Proposition~\ref{fact:plot},
consisting of $\Oh(b\log{\log{\lightweight(u_{i})}})$ bits.
Let $\seglenb(v_{i,j})$ be the size of the reserved segment of child $v_{i,j}$ after rounding in groups.
Then it can be seen that segment of length $1 + \sum_{j\geq 1} \seglenb(v_{i,j})$
will allow us to set the IDs for $u_{i}$ and subtrees rooted at its light children.
As the length of the considered heavy path is $p$,
then similarly with a segment of length $p + \sum_{i=1}^{p} \sum_{j\geq 1} \seglenb(v_{i,j})$
it will be possible to assign IDs in the whole $T_{u_1}$.

\begin{lemma}
$\seglen(u_1) = |T_{u_{1}}|2^{12\level(u_{1})/b}$ is enough for the encoder to properly assign IDs in the whole $T_{u_1}$.
\label{lem:span1}
\end{lemma}
\begin{proof}
By Lemma~\ref{lem:roundclasses}, Lemma~\ref{lem:potentials} and induction on the light depth of a node, we have:
\begin{align*}
p + \sum_{i=1}^{p} \sum_{j\geq 1} \seglenb(v_{i,j}) & \leq
p + \sum_{i=1}^{p} \sum_{j\geq 1} \seglen(v_{i,j})2^{11(\level(u_{1})-\level(v_{i,j}))/b} \\
& \leq p + \sum_{i=1}^{p} \sum_{j\geq 1} |T_{v_{i,j}}|2^{12\level(v_{i,j})/b}2^{11(\level(u_{1})-\level(v_{i,j}))/b}  \\
& \leq |T_{u_{1}}|2^{(12\level(u_{1})-1)/b}
\end{align*}
Then, another factor of $2^{1/b}$ factor is needed for rounding the value of $\bound$.
\end{proof}

Thus in the first phase, at the end of processing of a heavy path, $\seglen(u_1)$ is set to be just $|T_{u_{1}}|2^{12\level(u_{1})/b}$.
This way the bound from Claim~\ref{prop:size1} is satisfied.
Pseudocode for the first phase of the encoder is presented in Algorithm~\ref{alg:1phase}.

\begin{algorithm}[h]
\begin{algorithmic}[1]
  \Function{Assign-IDs}{$P, s$}
  \State \textbf{Input:} heavy path $P = u_1,\ldots,u_p$. $b$ is a fixed parameter.
  \State \textbf{Output:} $\start(u)$, $\bound(u)$ for every $u \in T_{u_1}$, $\start(u) \in [s,s + \seglen(u_1))$, $\nodespan(u_1) \leq \seglen(u_1)$.
  \State
  \State $A \gets s$ \Comment{Accumulator}
  \For{$i=1 \ldots p$}
    \State $\start(u_i) \gets A$
  	\State $A \gets A + 1$
    \For{$j=1 \ldots \degree(u_i)-1$} \Comment{Light children of $u_i$}
      \State \Call{Assign-IDs}{$P'$, $A$}, where $P'$ is a heavy path with $v_{i,j}$ as the head
      \State $A \gets A + \seglenb(v_{i,j})$
    \EndFor
  \EndFor
  \For{$i=1 \ldots p$}
  	\State Let $t$ be the smallest natural number such that $\lfloor 2^{t/b} \rfloor + \start(u_i) \geq A$
    \State $\bound(u_i) \gets \lfloor 2^{t/b} \rfloor$
  \EndFor
  \EndFunction
\end{algorithmic}
\caption{Second phase of the encoder, recursively assigning $\start$ and $\bound$ values.}
\label{alg:2phase}
\end{algorithm}

\paragraph{Second phase --- creating labels.}
The second phase of the encoder only unfolds the assigned $\seglenb$ into concrete values
of $\start$ and $\bound$ for every node.
It works in a top-down manner, using a recursive procedure scanning one heavy path in the tree at a time.
We start in the root with the accumulator $A$ set to $0$.
At node $u_{i}$ on a path, the procedure sets $\start(u_{i})$ to be the current value of the accumulator and then increases the accumulator by 1.
Then it iterates over the light children $v_{i,1},v_{i,2},\ldots$ of $u_{i}$.
For each $v_{i,j}$ we call the recursive procedure with a copy of the current accumulator as an argument,
and then increase the accumulator by $\seglenb(v_{i,j})$.
Finally, we need to set the proper $\bound(u_{i})$.
It can be seen that we only need that $\start(u_{i})+\bound(u_{i})$ is at least as large as the final accumulator.
Further, if the starting value of accumulator for a heavy path $P$ was $s$,
the final accumulator is $A = s + d$, where $d = \sum_{i=1}^{p} (1 + \sum_{j\geq 1}\seglenb(v_{i,j}))$.
Then, because $\start(u_i) \geq s$ and possible values of $\bound(u_{i})$ are all numbers of the form $\lfloor 2^{t/b}\rfloor$,
we can always choose $t$ so that $\start(u_{i}) + \bound(u_{i}) \leq s + d 2^{1/b}$.
At the end, by construction we obtain that $\nodespan(u_{1}) \leq \seglen(u_{1})$.
Note that in total we round the size of a given subtree twelve times for each passed
level: three times for rounding in classes to $\seglenp$,
(amortized) eight times for rounding in groups to $\seglenb$
and once when $\bound$ is increased to be power of $2^{1/b}$.
Pseudocode for the second phase of the encoder is presented in Algorithm~\ref{alg:2phase}.

\subsection{Resulting labels}
\paragraph{Final structure of a label.}
Label $\nodelabel(u)$ of node $u$ is composed of six parts: 
\begin{enumerate}
\item $\start(u)$,
\item $\bound(u)=\lfloor 2^{t/b}\rfloor$ stored as just $t$,
\item an encoding of $\rt(u)$,
\item $\lfloor \log{\lightweight(u)} \rfloor$,
\item $\level(u)$,
\item $\lceil \log{(\deg(u)-1)} \rceil$.
\end{enumerate}
$b$ will be computable from the known value of $\lceil \log{n} \rceil$.
The labels are distinct because $\start$ values are distinct.

\paragraph{Length of a label.}
From Lemma~\ref{lem:span1} we have $\nodespan(\treeroot(T)) \leq \seglen(\treeroot(T)) = 
n2^{12\log{n}/b} $.
Then every $\start$ value can be stored on $\log{n} + \Oh(\log{n}/b)$ bits, $\bound$ on
$\Oh(\log{(b\log{n})})$ bits and $\rt$ on $\Oh(b\log{\log{\lightweight(u_{i})}}) =  \Oh(b\log{\log{n}})$ bits.
Three last elements of a label can be stored on $\Oh(\log{\log{n}})$ bits.
By taking $b=\sqrt{\log{n}/\log{\log{n}}}$, which is at least $6$ for $n$ large enough,
we obtain labels of length $\log{n} + \Oh(\sqrt{\log{n}\log{\log{n}}})$.

\begin{figure}
  \includegraphics[scale=1.0]{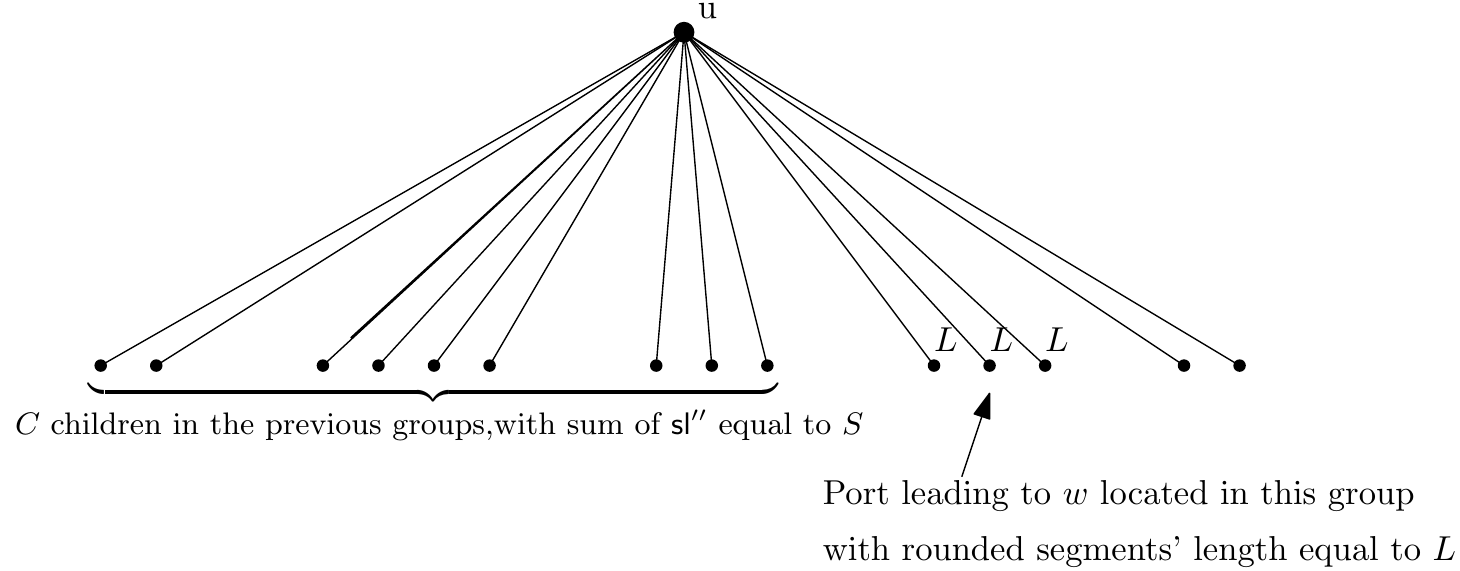}
  \caption{Information needed during a query: locating the correct group, the rounded size of a reserved segment there,
  and prefix sums of the number of ports and segment sizes in all previous groups.
  $\start(w) = \start(u) + S + kL + x$, and port $C+k+2$ is chosen.}
  \label{fig:querying}
\end{figure}

\subsection{Decoder}
The decoder can answer routing queries in polylogarithmic time, since the number of groups is polylogarithmic.
From $\nodelabel(u)$ and $\nodelabel(w)$ it can easily extract $\rt$, $\start$ and $\bound$ of both nodes.
If $\start(w) \not \in [\start(u),\start(u)+\bound(u))$, then $w \not \in T_{u}$ and port
number $0$ is chosen. Otherwise $\rt(u)$ is read and rounded $\seglenb$ of every light child of $u$ is retrieved.
We have groups of children of predefined sizes (depending only on $b$ and three last elements of a label),
with every child in a single group with the same $\seglenb$.
Let $v_1, v_2, \ldots$ be the light children of $u$ sorted by non-increasing size.
Then port $1$ leading to heavy child of $u$ is chosen if $\start(w) > \start(u) + \sum_{j=1} \seglenb(v_{j})$.
In the remaining case, port $i+1$ is chosen such that $i$ is smallest number for which $\start(w) \leq \start(u) + \sum_{j=1}^{i} \seglenb(v_{j})$ holds.
Port with a number greater than $\deg(u)$ will never be chosen.
Consult Figure~\ref{fig:querying} for the situation during a query.

\section{Final scheme with distribution of bits}
\label{sec:not_so_final}

In this section we design our final routing labeling scheme and prove the following theorem:

\begin{theorem}
There exists a canonical labeling scheme for routing in trees on $n$ nodes with labels of length $\log{n} + \Oh((\log{\log{n}})^2)$ bits,
the decoder answering queries in polylogarithmic time, and the encoder working in near-linear time.
\label{th:main}
\end{theorem}

The overall approach is as in the previous section, but now we choose
$b$ to not be the same value in every node, and move $\start$ values a little
to make this idea work efficiently.
First, notice that it might be excessive for $\rt(u)$ to always have the same amount of reserved bits.
Indeed, in most cases the light weight of a node is small,
and then we can get away with only a few bits for the routing table.
On the other hand, for nodes with a very significant light weight, it would be beneficial to use much more bits
to reduce the impact of rounding, since this rounding affects a huge part of the tree.
Thus, we use a simple method for compressing $\start$ values depending on the light weight of a given node,
known from other works in this area.

The idea is that if the binary expansion of $\start(u)$ happens to have $l$ trailing zeroes, then we
can save $l$ bits by removing these trailing zeroes, assuming that we are additionally storing their number.
Thus, we add $\Oh(\log\log n)$ bits to every label, but then are able to substantially shorten some of them.
The saved space can be used to store a relatively big routing table.
As we cannot be sure that binary expansions have a lot of trailing zeroes, we will enforce such a property.
If the light weight of $u$ is $\lightweight(u) = l$, we reserve an interval of
length $2^{\lceil \log{l} \rceil}$ exclusively for choosing the value $start(u)$.
In such an interval surely there is always a number with at least $\log{l}$ trailing zeroes,
so using that many bits for storing $rt(u)$ will be possible, which translates into
$b=\Omega(\log{l} / \log{\log{l}})$ for node $u$.
Moreover, the total size of intervals reserved this way is not too big, as every node contributes
to the light weight of at most $\log{n}$ ancestors, so the sum of intervals reserved for shifting $\start$ values will be at most $2n\log{n}$.
Note that equivalently we could store $\rt(u)$ as a suffix of the shifted value of $\start$,
instead of explicitly creating additional space by enforcing the existence of $\log{l}$ trailing zeroes,
but we find the latter solution cleaner for our purposes.

\subsection{Encoder}
\paragraph{First phase.}
For the first phase, we need to modify the previous bound for $\seglen$.

\begin{claim}
Under assumptions on $b$ and $\lightweight$ values made earlier,
using for every $u_1$ being the head of a heavy path
$\seglen(u_1)=2|T_{u_1}|\level(u_1)2^{\level(u_1)/\log{n}}\prod_{k=1}^{\level(u_1)}2^{14c\log{k}/k}$, 
for some big enough constant $c$, is sufficient for the encoder to achieve both $|\rt(u)| \leq \log{\lightweight(u)}$
and $\start(u)$ having $\log{\lightweight(u)}$ trailing zeroes for every node $u$.
\label{prop:size2}
\end{claim}

Claim~\ref{prop:size2} will be proved inductively with the modified procedure described below.
We need to make several changes to the encoder from Section~\ref{sec:improved2}.
Firstly, $b$ is no longer a global parameter, but rather a possibly different value for every node.
For a node $u$ we use $b=\log{\lightweight(u)}/c\log{\log{\lightweight(u)}}$, for some constant $c$.
The product in the bound for $\seglen(u)$ reflects such a choice of $b$, going over all levels up to $\level(u)$.
Note that rounding factors are decreasing, as the values of $b$ are increasing and thus $2^{1/b}$ are decreasing.
We require that $b \geq 6$,
meaning that for every node $\log{\lightweight(u)} / \log{\log{\lightweight(u)}} \geq 6c$,
which is realised if we assume that $\lightweight(u)$ is big enough.
Secondly, we increase the range for $\start$ value from just $1$ to $2\lightweight(u)$,
ensuring that the assignment of $\start(u)$ with sufficient number of trailing zeroes will be possible.

Thirdly, we need to adjust $\bound$ rounding.
As the value of $b$ is now set locally for a node $u$, depending solely on its light weight,
we can no longer round up $\bound(u)$ by a factor of $2^{1/b}$, since this rounding
applies to more than just the subtrees of light children of $u$.
Because of that, now for every node we just round $\bound$ by a factor $2^{1/\log{n}}$.
This is a minor issue, as storing $\bound$ was never main factor in the length of a label, and
such modified values of $\bound$ will still be stored on just $\Oh(\log{\log{n}})$ bits.
Lastly, we need to revisit rounding in classes.
Rounding in groups works exactly as before.
However, for classes, because of the new formula for length of segments,
now we have different boundary values than in Section~\ref{sec:improved2}.
Namely, if a class consists of every light child $v_{i,j}$ with $|T_{v_{i,j}}| \in [x_1,x_2+1)$,
then the boundary value for this class will be $2 x_2 \log{x_2} 2^{\log{x_2}/\log{n}}\prod_{k=1}^{\log{x_2}}2^{14c\log{k}/k}$.
After this change we have a bound analogous to Lemma~\ref{lem:roundclasses}:

\begin{lemma}
\label{lem:roundclasses2}
Consider node $u_i$ for which we have $b_i=\log{\lightweight(u_i)}/c\log{\log{\lightweight(u_i)}}$,
and let $l = \min(\lfloor \log{\lightweight(u_{i})} \rfloor + 1, \level(u_i))$.
If $v_{i,j}$, a light child of $u_i$ with $\level(v_{i,j}) = l - k$, is part of class $C$,
then the boundary value for this class is no larger than $\seglen(v_{i,j})2^{6k/b_i}$.
\end{lemma}
\begin{proof}
If $C$ is one of the classes subdividing some preclass of rank at most $b_i-1$,
then as before the boundary value for this class is actually no larger than $\seglen(v_{i,j})2^{k/b_i}$,
as interval of $C$ spans inside just a single level, and this level was subdivided into $b_i/k$ classes.
Now assume that $C$ is a class constructed by merging $r$ preclasses of rank between $rb_i$ and $2rb_i-1$.
Then $\level$ of two children in $C$ must differ by less than $r$, which means that their 
subtrees' sizes differ by less than a factor $2^r$.
Thus, the boundary value for $C$ is less than $\seglen(v_{i,j})2^r(1+1/b_i)2^{r/\log{n}}\prod_{z=\level(v_{i,j})}^{\level(v_{i,j})+r-1} 2^{14c\log{z}/z}$.
By assumption on the light weight of nodes, which guarantees that $b$ is always at least $6$,
and as $\level(u)$ is not less than $\log{\lightweight(u)}$,
$\prod_{z=\level(v_{i,j})}^{\level(v_{i,j})+r-1} 2^{14c\log{z}/z} \leq \prod_{z=\level(v_{i,j})}^{\level(v_{i,j})+r-1} 2^{14/6} \leq 2^{14r/6}$. 
Furthermore, $(1+1/b_i) < 2$ and $2^{r/\log{n}} < 2$ as $r < \log{n}$.
Therefore the boundary value of $C$ is less than $\seglen(v_{i,j})2^{r+2+14r/6} < \seglen(v_{i,j})2^{6r}$.
As $C$ is constructed by merging preclasses of rank at least $rb_i$,
then $\level(v_{i,j}) \leq l - rb_i$, so we have that $k \geq rb_i$ and
finally the boundary value for $C$ is at most $\seglen(v_{i,j})2^{6k/b_i}$.
\end{proof}

\begin{algorithm}
\begin{algorithmic}[1]
  \Function{Assign-$\seglen$}{$P$}
  \State \textbf{Input:} heavy path $P = u_1,\ldots,u_p$.
  \State \textbf{Requirement:} all nodes in the subtrees hanging off $P$ are already processed.
  \State \textbf{Output:} $\rt(u_i)$ and $\seglenb(v_{i,j})$ values for every $u_{i}$.
  \State $c$ --- fixed constant
  \State
  \State $r \gets 0$
  \For{$i=1 \ldots p$}	
  	\State \textcolor{blue}{$b_i \gets \log{\lightweight(u_i)}/c\log{\log{\lightweight(u_i)}}$}
  	\Comment{Dynamic value of $b$}
    \State $C \gets$ \Call{Construct-classes}{$u_i$, $b_i$}
    \For{$j=1 \ldots \degree(u_i)-1$}
      \State Set $\seglenp(v_{i,j})$ to be the boundary value of the class of $v_{i,j}$ in $C$
    \EndFor
	\State \textcolor{blue}{$segment_i \gets 2\lightweight(u_i)$}
	\Comment{Interval reserved for $\start(u_i)$}
    \State $G \gets$ \Call{Construct-groups}{$u_i$, $b_i$}
    \For{$j=1 \ldots \degree(u_i)-1$}
      \State Set $\seglenb(v_{i,j})$ to be size of the largest child in group of $v_{i,j}$ in $G$
      \State $segment_i \gets segment_i + \seglenb(v_{i,j})$
    \EndFor
    \State $r \gets r + segment_i$    
    \State Set $\rt(u_i)$ as the description of rounded groups
  \EndFor
  \State \textcolor{blue}{\textbf{Guaranteed property:} $r \leq 2|T_{u_1}|\level(u_1)2^{(\level(u_1)-1)/\log{n}}\prod_{k=1}^{\level(u_1)}2^{14c\log{k}/k}$}
  \State \textcolor{blue}{$\seglen(u_1) \gets 2|T_{u_1}|\level(u_1)2^{\level(u_1)/\log{n}}\prod_{k=1}^{\level(u_1)}2^{14c\log{k}/k}$}
  \EndFunction
\end{algorithmic}
\caption{First phase of the modified encoder, assigning $\seglenb$ and $\rt$ on a heavy path.
We do not include the change to the boundary values in function \textsc{Construct-classes} there.}
\label{alg:1phaseagain}
\end{algorithm}

Pseudocode for the first phase of the modified encoder is presented in Algorithm~\ref{alg:1phaseagain}.
It can be seen that with a segment of length $\sum_{i=1}^{p} (2\lightweight(u_i) + \sum_{j\geq 1} \seglenb(v_{i,j}))$,
it will be possible for the encoder to properly assign IDs in the whole $T_{u_1}$.
We need a bound on this value.

\begin{lemma}
$\seglen(u_1) = 2|T_{u_1}|\level(u_1)2^{\level(u_1)/\log{n}} \prod_{k=1}^{\level(u_1)}2^{14c\log{k}/k}$
is enough for the encoder to properly assign IDs in the whole $T_{u_1}$.
\label{lem:size2}
\end{lemma}
\begin{proof}
By Lemma~\ref{lem:roundclasses2}, Lemma~\ref{lem:potentials} and induction on the light depth of a node.
\newpage Let $l_i = \min(\lfloor \log{\lightweight(u_i)} \rfloor + 1, \level(u_i))$,
and $b_i = \log{\lightweight(u_i)}/c\log{\log{\lightweight(u_i)}}$, then we have:
\begin{align*}
& \sum_{i=1}^{p} (2\lightweight(u_i) + \sum_{j\geq 1} \seglenb(v_{i,j})) 
 \leq \sum_{i=1}^{p} (2\lightweight(u_i) + \sum_{j\geq 1} \seglen(v_{i,j})2^{14(l_i-\level(v_{i,j}))/b_i}) \\
&\leq 2|T_{u_{1}}| + \sum_{i=1}^{p} \sum_{j\geq 1} 2|T_{v_{i,j}}|\level(v_{i,j})2^{\level(v_{i,j})/\log{n}}2^{14(l_i-\level(v_{i,j}))/b_i}\prod_{k=1}^{\level(v_{i,j})}2^{14c\log{k}/k} \\
&\leq 2|T_{u_{1}}| + (\level(u_1)-1) 2^{(\level(u_1)-1)/\log{n}} \sum_{i=1}^{p} \sum_{j\geq 1} 2|T_{v_{i,j}}| \prod_{k=1}^{l_i} 2^{14c\log{k}/k} \\
& \leq 2|T_{u_1}|\level(u_1)2^{(\level(u_1)-1)/\log{n}} \prod_{k=1}^{\level(u_1)}2^{14c\log{k}/k}
\end{align*}
Then, increase by a single factor of $2^{1/\log{n}}$ is needed for $\bound$ rounding.
\end{proof}

Thus in the first phase, at the end of processing of a heavy path, $\seglen(u_1)$ is set as follows as to satisfy the bound from 
Claim~\ref{prop:size2} :
\[ 2|T_{u_1}|\level(u_1)2^{\level(u_1)/\log{n}} \prod_{k=1}^{\level(u_1)}2^{14c\log{k}/k} . \]

\begin{algorithm}[h]
\begin{algorithmic}[1]
  \Function{Assign-IDs}{$P,s$}
  \State \textbf{Input:} heavy path $P = u_1,\ldots,u_p$. $c$ is a fixed constant, $\log{n}$ is known. 
  \State \textbf{Output:} $\start(u)$, $\bound(u)$ for every $u \in T_{u_1}$, $\start(u) \in [s,s + \seglen(u_1))$, $\nodespan(u_1) \leq \seglen(u_1)$.
  \State
  \State $A \gets s$
  \For{$i=1 \ldots p$}
    \State \textcolor{blue}{$b \gets \log{\lightweight(u_i)}/c\log{\log{\lightweight(u_i)}}$}
  	\State \textcolor{blue}{$A \gets A + 2\lightweight(u_i) - 1$}
  	\Comment{Reserving an additional interval}
  	\State \textcolor{blue}{$l \gets 2^{\lceil \log{\lightweight(u_i)} \rceil}$}
    \State \textcolor{blue}{$\start(u_i) \gets \lfloor A / l \rfloor \cdot l$}
    \Comment{$\start$ value gets enough trailing zeroes}
    \State \textcolor{blue}{$A \gets \start(u_i)+1$}
    \For{$j=1 \ldots \degree(u_i)-1$}
      \State \Call{Assign-IDs}{$P'$, $A$}, where $P'$ is a heavy path with $v_{i,j}$ as the head
      \State $A \gets A + \seglenb(v_{i,j})$
    \EndFor
  \EndFor
  \For{$i=1 \ldots p$}
  	\State \textcolor{blue}{Let $t$ be the smallest natural number such that $\lfloor 2^{t/\log{n}} \rfloor + \start(u_i) \geq A$}
    \State \textcolor{blue}{$\bound(u_i) \gets \lfloor 2^{t/\log{n}} \rfloor$}
  \EndFor
  \EndFunction
\end{algorithmic}
\caption{Second phase of the modified encoder, recursively assigning $\start$ and $\bound$ values.}
\label{alg:2phaseagain}
\end{algorithm}

\paragraph{Second phase --- creating labels.}
The second phase is very similar to the one from Section~\ref{sec:improved2}.
We start in the root with the accumulator set to $0$, and then process heavy paths in a top-down manner.
At node $u_i$ on a heavy path and with the accumulator being $A$, the procedure sets $\start(u_{i})$ to be the greatest value with at least $\lceil \log{\lightweight(u_i)} \rceil$ trailing zeroes not bigger than $A+2\lightweight(u_i)-1$, and then sets $A=\start(u_{i})+1$.
Next, we iterate over the light children $v_{i,1},v_{i,2},\ldots$ of $u_{i}$.
For each $v_{i,j}$ we call the recursive procedure with a copy of the current accumulator as an argument,
and then increase the accumulator by $\seglenb(v_{i,j})$.
Finally, we need to set proper $\bound(u_{i})$ for every $i=1,2,\ldots,p$,
so that $\start(u_{i})+\bound(u_{i})$ is at least as large as the final accumulator.
If the initial value of the accumulator for a heavy path $P$ is $s$,
then the final accumulator is $A = s + d$, where $d \leq \sum_{i=1}^{p} (2\lightweight(u_i) + \sum_{j\geq 1}\seglenb(v_{i,j}))$.
Finally, because $\start(u_i) \geq s$  and the possible values of $\bound(u_{i})$ are all numbers of the form $\lfloor 2^{t/\log{n}}\rfloor$,
we can always adjust $t$ so that $\start(u_{i}) + \bound(u_{i}) \leq s + d 2^{1/\log{n}}$.
Thus, by construction we obtain that $\nodespan(u_{1}) \leq \seglen(u_{1})$.
See Algorithm~\ref{alg:2phaseagain}.

\subsection{Resulting labels}
\paragraph{Final structure of a label.}
As in Section~\ref{sec:improved2}. Note that the stored $\log{\lightweight(u_i)}$ allows recovering $b$.
The number of trailing zeroes in $\start(u)$ is saved separately, and these zeroes are cut off
from the binary representation.

\paragraph{Length of a label.}
By Lemma~\ref{lem:size2} we have $\nodespan(\treeroot(T)) \leq \seglen(\treeroot(T)) = 
4n\log{n} \prod_{k=1}^{\log{n}}2^{14c\log{k}/k}$.
As we are interested in the logarithm of this value that affects the length of $\start$, we calculate:
\begin{align*}
\log{(\prod_{k=1}^{\log{n}}2^{14c\log{k}/k})} = \Oh(\sum_{k=1}^{\log{n}}\log{k}/k) = \Oh((\log{\log{n}})^2).
\end{align*}
Note that the rounding is not evenly distributed among levels. Inside small subtrees, where we cannot provide
many trailing zeroes, the values of $b$ are relatively small and thus roundings are quite large deep in the tree.
Regardless of that, every $\start$ could be stored on $\log{n} + \Oh((\log{\log{n}})^2)$ bits, but
every $\start(u)$ is guaranteed to have $\log{\lightweight(u)}$ trailing zeroes,
so we actually store this value on $\log{n} + \Oh((\log{\log{n}})^2) - \log{\lightweight(u)}$ bits.
Clearly, $|\bound(u)| = \Oh(\log{\log{n}})$.
For $b$ values, we use $c=2$, or rather exactly $b=\log{\lightweight(u)}/2(\log{\log{\lightweight(u)}}+3)$.
As the number of bits used to describe the routing table is at most twice the number of groups, and there are at most 
$b(\log{\log{\lightweight(u)}}+3)$ of them, $|\rt(u)| \leq \log{\lightweight(u)}$.
Overall, the labels consist of $\log{n} + \Oh((\log{\log{n}})^2)$ bits.

\subsection{Decoder}
\paragraph{Decoder.}
Decoding proceeds as in Section~\ref{sec:improved2}, after retrieving the value of $b$.

\paragraph{Light weight assumption.}
At the beginning of Section~\ref{sec:improved2}, we assumed that for every node $u$ $\lightweight(u) \geq c'$ for some constant $c'$.
As hinted, it could be realised in the following way.
Firstly, the encoder adds $c'+1$ artificial leaves to every node of the input tree.
Then the assumption holds for every original node, but the encoder needs to handle these additional nodes.
In a heavy path decomposition, artificial nodes fall into two categories.
Most of them are heads of heavy paths of length one, and remaining ones are bottom nodes
at the ends of some longer heavy paths.
Both can be handled very similarly, as they have no children to process.
For leaves, in the Algorithm~\ref{alg:1phaseagain}, $\rt(u)$ is set to null 
(we can use additional single bit in the label indicating whether $u$ does have any children),
and if an artificial node $u$ is a head of a heavy path then the encoder sets $\seglen(u)=1$.
In Algorithm~\ref{alg:2phaseagain}, ID of an artificial node $u$ is set to be
the current value of the accumulator ($\start(u) \gets A$), $\bound(u)$ is set to $1$,
and the accumulator is increased by one.

As we increase size of an input tree by a constant factor and our labeling scheme uses labels of length $\log{n}+\Oh((\log{\log{n}})^2)$,
adding artificial children increases the length of a label only by $\Oh(1)$ bits.
Moreover, as a labeling scheme is canonical, ports leading to the artificial children of a node $u$
are $\deg(u)+1,\ldots,\deg(u)+c'$.
Thus, if as the last step the encoder disregards the labels of artificial nodes and edges,
then we are left with a correct labeling for the input tree.
The decoder does not change.
Note that the concept of adding artificial leaves is used purely for simplicity of the analysis.

\section{Simple extensions}
\label{sec:extensions}
In Section~\ref{sec:preliminary} we presented a simple routing labeling scheme for bounded degree trees
with better lengths than general scheme.
Another extension is the case of a small depth.

\begin{corollary}
For trees of bounded depth there is a routing labeling scheme with labels of length $\log{n}+\Oh(\log{\log{n}})$ bits.
\end{corollary}
\begin{proofs}Only cause of label size exceeding $\log{n}+\Oh(\log{\log{n}})$ for the scheme 
from Section~\ref{sec:not_so_final} is length of $start(u)$, enforced by rounding accumulated over levels.
But for trees of bounded depth, as $b$ is always at least 1,
segments' lengths in the described scheme can be adjusted to constraint overall roundings to a factor of $\Oh(2^{d})$, where $d$ is the depth of a tree.
\end{proofs}

\paragraph{Routing with larger local memory.}
A labeling scheme provides a symmetric solution for the routing problem while using little space, in which every node of the network
has an assigned label and during a routing query the decoder gets just labels of the current node and the destination node.
A label of the destination node serves as its address and travels through the network, in the header attached to a packet,
and the label of a node is stored in its a local memory as its identifier.
Typically, relaxation of this setting is used, allowing nodes to store locally not only a small label, but (much) more bits.
In such case, a further decrease in length of a label (address) attached to a packet is possible.

\begin{definition}
A routing labeling scheme \emph{with local tables} for a family of rooted trees $\mathcal{T}$ consists of an encoder and a decoder.
The encoder takes a tree $T\in \mathcal{T}$, then assigns a label $\nodelabel(u)$ and a local table $\lt(u)$ to every node $u\in T$.
Edges are numbered with port numbers as in the regular designer-port routing labeling scheme.
The decoder receives $\lt(u)$ and $\nodelabel(w)$, such that $u,w\in T$ for some $T\in \mathcal{T}$ and $u \not = w$.
The decoder should return the port number corresponding to the first edge on the path from $u$ to $w$. 
We are interested in minimising the maximum length of a label, and minimising the maximum length
of the local table as a secondary objective.
\end{definition}

Our main result gives a routing labeling scheme with both local tables and labels of size $\log{n}+\Oh((\log{\log{n}})^2)$.
Here we give two other trade-offs further decreasing size of a label at the expense of increasing local space used.

Note that routing labeling scheme from Section~\ref{sec:improved2} possibly makes use of every part of label $\nodelabel(u)$,
but checks only the value of $\start(w)$ from $\nodelabel(w)$.
Thus we can construct a routing labeling scheme with local tables by first running the encoder 
described in Section~\ref{sec:improved2} and then, for every node $u$,
setting $\lt(u)$ to be returned label of $u$ and storing as shrunk $\ell(u)$ only $\start(u)$.
Recall that this labeling scheme does not use hiding information by compressing trailing zeroes.
Additionally, there are no parts of the new label to be separated, so we can just set $\nodelabel(u)=\start(u)$.
This way, setting $b=\log{n}/\log{\log{n}}$, we achieve $\nodespan(\treeroot(T)) \leq n2^{12\log{n}/b}$, and the size of $\rt$ is $\Oh(\log{n})$.

\begin{corollary}
There is a routing labeling scheme with local tables having label length bounded by $\log{n}+\Oh(\log{\log{n}})$ and local tables size bounded by $\Oh(\log{n})$.
\label{cor:loc1}
\end{corollary}

\noindent Similarly, setting $b=\log{n}$ we get the following:

\begin{corollary}
There is a routing labeling scheme with local tables having label length bounded by $\log{n}+\Oh(1)$ and local tables size bounded by $\Oh(\log{n}\log{\log{n}})$.
\label{cor:loc2}
\end{corollary}

\section{Lower bound of $\log{n}+\Omega(\log{\log{n}})$}
\label{sec:lower_bound}
We first describe the well-known lower bound for ancestry labeling schemes, as the
lower bound for routing labeling schemes is just its minor adjustment.
We fix number of nodes, $n$, and define $T_i$ to be a tree consisting of $i$ paths of roughly
equal length hanging off a root.

Firstly, observe that all labels must be different, because every node is its own ancestor, and $u=w$ is the only case where the
answer for both queries $(\ell(u),\ell(w))$ and $(\ell(w),\ell(u))$ is positive.
We want to consider the assignment of labels by the encoder for $T_1,T_{2},\ldots$,
keeping track of the number of distinct labels that have already been used on some nodes.
For $T_1$ there is one path, and every node has a different label.
We count and mark all these labels.
For $T_2$, we want to argue that already marked labels may appear only on one of the two paths.
Indeed, assume that two marked (and necessarily different) labels $\ell_1$ and $\ell_2$ appeared on different paths.
But as they are marked, they appeared on one path in $T_1$, which means that the decoder has to
answer positively ancestry query for these two labels, in one or another order.
This is a contradiction, as the decoder must not answer positively ancestry query for two nodes from different paths.
Thus, one of the paths in $T_2$ must have been assigned only unmarked labels.
We count and mark all these new labels.
In general, considering $T_i$, there were labels from $i-1$ paths counted and marked in the previous trees.
This means that the encoder has to assign unmarked labels to at least one of the paths in $T_i$.
We choose one such a path, then count and mark labels from it.
At the end there are $\sum_{i=1}^{n} \lfloor n/i \rfloor = \Omega(n\log{n})$ different marked labels,
which leads to the desired lower bound on length of a label.

To move to a lower bound for routing, we just note that in considered trees every node
except the root has degree of at most 2.
Thus, we can artificially augment the labels created by a labeling scheme for routing in such trees, 
adding a single bit denoting which of at most two ports leads to a parent of a node.
This effectively allows answering ancestry queries, as $u$ is an ancestor of $w$ iff decoder is
answering a query $(\ell(u),\ell(w))$ with a port not leading to the parent of $u$.
If the decoder for a routing labeling scheme does not have to provide answer to queries in form of $(u,u)$,
we need some other way to enforce uniqueness of labels on paths.
It can be done by attaching a single dummy node to every node on a path.
Finally, we notice that we could use trees with degree two, by expanding high-degree root into a binary tree,
adding only a linear number of new vertices.
This gives us the following:

\begin{theorem}
Any labeling scheme for routing in trees on $n$ nodes, even with degree bounded by 2,
needs labels consisting of $\log{n} + \Omega(\log{\log{n}})$ bits.
\end{theorem}

\section{Conclusions}

We have designed a labeling scheme for routing in trees on $n$ nodes with labels of length
$\log{n} + \Oh((\log{\log{n}})^2)$. While this is a major step in determining the asymptotically
correct second-order term, we still have a gap between the (simple) lower bound of $\log{n}+\Omega(\log{\log{n}})$
and our upper bound of $\log{n} + \Oh((\log{\log{n}})^2)$. It doest not seem possible to bridge this
gap by working in our framework, as it appears that it incurs a multiplicative penalty of
$\log\log n'$, for every $n'=n,n^{1/2},n^{1/4},\ldots$, in the second-order term. It seems to
us that routing is harder than ancestry, and so $\Omega((\log\log n)^{2})$ might be the
right answer. However, showing this requires fixing a family of trees that are complicated
but at the same time possible to analyse. We have shown that trees of bounded degree or depth
admit a scheme with the second-order being $\Oh(\log{\log{n}})$, which suggests that this
family should consists of trees with logarithmic depth and many nodes having large ``entropy''
of (rounded) sizes of their subtrees. This seems hard to quantify and analyze,
and the bound of $\log{n} + \Oh(\log{\log{n}})$ in the model with local tables
of size $\Oh(\log{n})$ shows that some natural approaches cannot work.

\bibliographystyle{plain}
\bibliography{biblio}

\appendix
\newpage

\section{Constant time query}
\label{sec:const}

\paragraph{Technical aspects of roundings.}
In our labeling schemes we frequently used rounding up to $\lfloor 2^{t/b}\rfloor $, for integer $t$ and some fixed integer parameter $b$.
Numbers from range $[0 , n-1]$ can be stored, after such rounding, on $\log{(b\log{n})}=\log{b}+\log{\log{n}}$ bits.
This is convenient for the analysis, but not necessarily so for implementing operations such as taking powers or logarithms.
As a substitute, one could use a very similar rounding, which we call a \emph{two-parts representation}, that also
turns out to be more useful if we want to implement queries in constant time.
Instead of rounding up given number $x \in [0 , n-1]$ to the power of $2^{1/b}$,
we can store the first $\log{b}+2$ most significant bits of $x$ and also the length
of its binary representation on $\log{\log{n}}$ bits.
This way only $\log{b}+2+\log{\log{n}}$ bits are used, and bits of $x$ after the $(\log{b}+2)$-th one are lost
(we round up by just adding $1$ to the stored most significant bits),
which results in rounding by a factor of $1+1/2b$.
As $1+1/2b \leq e^{1/2b} < 2^{1/b}$, this is not more than in the previous method.
But now, two-parts representation operates only on powers and logarithms in base two.
This representation is basically floating-point numbers with precision parameterized by $b$,
used for integers only, and always rounding up.

\paragraph{Finding port number with prefix sums.}
This appendix gives an overview of a routing labeling scheme with a constant-time decoder, providing description of a construction and sketched proofs.
The additional ingredient is a data structure for storing some information concerning prefixes of children. Our goal will be to prove the following
theorem:

\begin{figure}
  \includegraphics[scale=1.0]{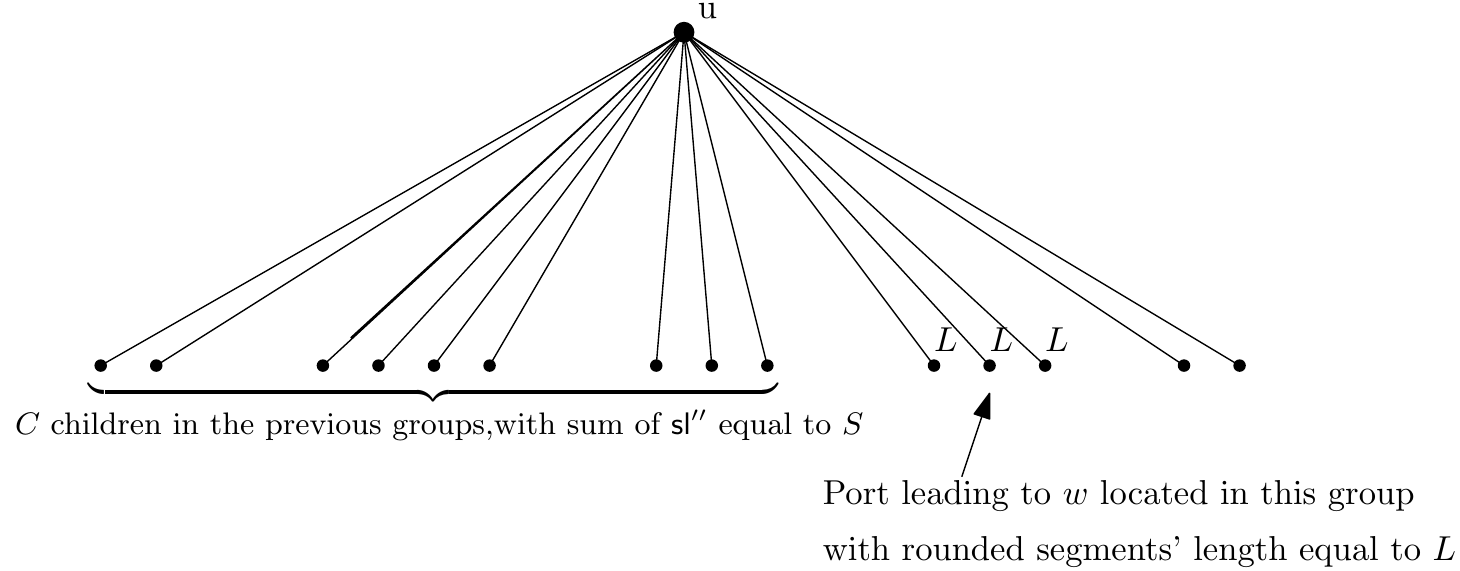}
  \caption{Information needed during a query: locating the correct group, the rounded size of a reserved segment there,
  and prefix sums of the number of ports and segment sizes in all previous groups.
  $\start(w) = \start(u) + S + kL + x$, and port $C+k+2$ is chosen.}
  \label{fig:querying2}
\end{figure}

\begin{theorem}
There exists a labeling scheme for routing in trees on $n$ nodes with labels of length $\log{n} + \Oh((\log{\log{n}})^3)$ bits,
the decoder answering queries in constant time, and the encoder working in polynomial time.
\label{th:const}
\end{theorem}

In Section~\ref{sec:not_so_final}, we had to decode the routing table bit by bit, which prevented us from answering queries in
constant time. Recall that we only need to consider the case that the decoder learned from $\ell(u)$ and $\ell(w)$ that
$w$ is in a subtree rooted at a light child $u$.
Let $v_w$ be the child of $u$ which is also an ancestor of $w$.
Children of $u$ were divided into groups, and children in the same group have equal reserved segment length ($\seglenb$).
Now four pieces of information are needed to locate the right port leading to $v_w$ from $u$:
\begin{itemize}
\item In which group of children of $u$ $v_w$ is; only as leading to the following values, the number of a group in itself is not interesting.
\item $L$, rounded segment length in this group.
\item $S$, prefix sum of reserved segments of all children in the prior groups.
\item $C$, the number of all children in the prior groups.
\end{itemize}
Consult Figure~\ref{fig:querying2}.
Then, as port $0$ leads to the parent of $u$ and $1$ to its heavy child,
the decoder needs to answer with a port numbered $1+C+ \lceil (\start(w)-\start(u)-S)/L \rceil$.

If a scheme is able to store prefix sums for lengths of reserved segments and sizes of groups (measured in the number of children),
then the decoder can use either binary search for locating the correct port or better use some static data structure allowing for
constant time $\rank$ queries, where for a fixed set $\rank(x)$ returns number of elements from this set less than $x$.
We use a structure based on a parallel comparison method from \cite{FusionTrees} to be described later.
\begin{lemma}
For any positive integers $w'$ and $k$ such that $w'k = \Oh(w)$, it is possible to store $k$ numbers of binary length
at most $w'$ on $\Oh(w'k)$ bits, while supporting $\rank$ queries in constant time.
Note that we might be using space smaller than one full machine word, so just its fragment.
\label{lem:rank2}
\end{lemma}
We will refer to this structure as a \emph{dictionary}.
Effectively, we use $\rank$ queries just to locate predecessors in constant time.

\paragraph{Creating prefix sums.}
The decoder from Section~\ref{sec:not_so_final}, in some sense, used just one bit in $\rt$ for every group.
Simple prefix sums need to be longer, thus slight increase in the label length.
Sums should not be too large, though --- we use $\Oh(\log{\log{n}})$ bits for each.
This will introduce yet another rounding for $\nodespan$ of nodes, which needs to be analysed.

Assume in advance that we are able to keep the size of the reserved segment for a node $u_i$ within
$\Oh(|T_{u_i}|\level(u_i)2^{\Oh((\log{\log{\lightweight(u_i)}})^3)})$ bound, so besides $\start(u_i)$ and $\bound(u_i)$
the size of binary representation of numbers in use is $\Oh(\log{\lightweight(u_i)})$,
as they depend only on the total size of subtrees of light children of $u_i$.
Then, we need to store the (rounded) prefix sums $\sum_{j=1}^{k} \seglenb({v_{i,j}})$ ---
and $\Oh(b\log{\log{\lightweight(u_i)}})$ of them, according to the division into groups;
recall there are $\Oh(b\log{\log{\lightweight(u_i)}})$ groups.
In other words, we want to ensure that the sum of the lengths of reserved segments for children from the first group is rounded,
then that the sum of the lengths of reserved segments for children from the first two groups is rounded, and so on.
Note that we round prefix sums, not individual segment's length, and increase of the length of a single
segment might be very significant --- in the worst case rounding is applied many times to 
almost the whole segment reserved for $u_1$.

As we plan to stick to rounding by a factor of $2^{\Oh(1/b)}$ per level,
we can afford rounding by a factor of $2^{\Oh(1/b^2\log{\log{\lightweight(u_i)}})}$ for every created prefix sum.
Thus, these prefix sums can be stored in rounded two-parts representation, in the form of $\Oh(\log{b}+\log{\log{\log{\lightweight(u_i)}}})$
most significant bits and then $\Oh(\log{\log{\lightweight(u_i)}})$ bits encoding the number of following zeroes.
With this many bits used for every prefix sum, so for every group,
whole description of $\rt$ takes $\Oh(b(\log{\log{\lightweight(u_i)}})^2 + b\log{b}\log{\log{\lightweight(u_i)}})$ bits,
when Lemma~\ref{lem:rank2} is used.
Note that we have two-parts representation of numbers, with some most significant bits and then a number of following zeroes,
while a dictionary from Lemma~\ref{lem:rank2} operates just on the usual binary representation of numbers.
We will deal with this small issue later, for now, let us assume it can be overcome.

By setting $b=\log{\lightweight(u_i)}/c(\log{\log{\lightweight(u_i)}})^2$ for some constant $c$,
we would fit the whole $\rt(u_i)$ in the $\log{\lightweight(u_i)}$ available bits created during encoding
by storing the trailing zeroes of $\start(u_i)$ separately.
Then, as $b$ is divided by an additional $\log{\log{\lightweight(u_i)}}$
(recall that in Section~\ref{sec:not_so_final} $b=\log{\lightweight(u_i)}/c\log{\log{\lightweight(u_i)}}$ was used),
it can be checked with similar inequalities bounding $\seglen$ as before,
only taking into the account the constant number of additional $2^{1/b}$ factors for every level from creating prefix sums,
that we will get $\log{n}+\Oh((\log{\log{n}})^3)$ bits as the length of a label.

But, as additional information, we need prefix sums on groups' sizes to know how many ports need to be skipped.
Analogically, we can afford rounding by a factor of $2^{\Oh(1/b^2\log{\log{\lightweight(u_i)}})}$ for every group.
Therefore, rounding prefix sum at every group by storing $\Oh(\log{b}+\log{\log{\log{\lightweight(u_i)}}})$ most significant bits
and then the number of following zeroes is going to be sufficient ---
total increase in reserved segments' length will be just a factor of $2^{\Oh(1/b)}$,
and such a prefix sum takes $\Oh(\log{b}+\log{\log{\lightweight(u_i)}})$ bits to be stored.
As we are increasing the sizes of groups, we do not want to add dummy nodes to any group but the last one,
so rounding of the number of children is done by increasing the reserved segments for a necessary number of children
in the following groups.
More precisely, after rounding the size of a group $g$ up to $z$, until $g$ consists of exactly $z$ light children of $u_i$,
the reserved segment of a single child in any of the further groups is artificially increased to the size of segments in $g$,
and then this child is moved to $g$.
This way, dummy nodes will have to be used for the last group only, which cannot borrow children from the further groups,
but in this case it is not an issue, as the decoder will never answer a query
with a port leading to a dummy child inside the last group.

The total increase in $\seglen(u_i)$ is indeed by a factor of $2^{\Oh(1/b)}$,
as the already existing groups were increased by at most this value.
These prefix sums are built in parallel to prefix sums of segment lengths,
so during processing of a single group first we round prefix sum on groups' sizes,
then prefix sum on segment lengths, then proceed to the next group.

Finally, together with the prefix sum for the number of children, we also need to store the size of reserved segments in a given group.
Rounded by $2^{\Oh(1/b)}$, it needs $\Oh(\log{b}+\log{\log{\lightweight(u_i)}})$ bits
to be stored in two-parts representation.

To sum up, the entry for a group $g_j$ consists of three elements: 
\begin{itemize}
\item As a key for the dictionary, value $S_j$ being the rounded sum of all segments used for groups $1, \ldots, (j-1)$.
Then the reserved segment of the first child in $g_j$ is starting at $\start(u_i)+S_j+1$.
$S_j$ is stored in rounded two-parts representation, with $\Oh(\log{b})$ most significant bits saved,
which takes space of $\Oh(\log{b}+\log{\log{\lightweight(u_i)}})$ bits. 
\item $C_j$, number of children of $u_i$ already processed in the previous groups.
Through moving children between groups this number was made to have just $\Oh(\log{b})$ significant bits, and then all zeroes,
thus taking $\Oh(\log{b}+\log{\log{\lightweight(u_i)}})$ bits to store exactly.
\item $L_j$, the size of the reserved segments in $g_j$, also rounded and stored on $\Oh(\log{b}+\log{\log{\lightweight(u_i)}})$ bits.
\end{itemize}
These three values enable locating the sought port number in a way described at the beginning.
We make them accessible through $\rank$ queries on $S_j$ values --- the answer to this query is interpreted as an index in an array,
and there three values are stored together.
As any of these values for a single group takes $\Oh(\log{b}+\log{\log{\lightweight(u_i)}})$ bits,
we still can use $b=\log{\lightweight(u_i)}/c(\log{\log{\lightweight(u_i)}})^2$ to achieve
labels with length $\log{n}+\Oh((\log{\log{n}})^3)$ bits.

\begin{algorithm}[h]
\begin{algorithmic}[1]
  \Function{Assign-$\seglen$}{$P$}
  \State \textbf{Input:} heavy path $P = u_1,\ldots,u_p$. $c$ is a fixed constant.
  \State \textbf{Requirement:} all nodes in the subtrees hanging off $P$ are already processed
  \State
  \For{$i=1 \ldots p$}	
  	\State $b_i \gets \log{\lightweight(u_i)}/c(\log{\log{\lightweight(u_i)}})^2$
    \State $C \gets$ \Call{Construct-classes}{$u_i$, $b_i$}
    \For{$j=1 \ldots \degree(u_i)-1$}
      \State Set $\seglenp(v_{i,j})$ to be the boundary value of the class of $v_{i,j}$ in $C$
    \EndFor
    \State $G \gets$ \Call{Construct-groups}{$u_i$, $b_i$}
    \For{$j=1 \ldots \degree(u_i)-1$}
      \State Set $\seglenb(v_{i,j})$ to be the size of the largest child in the group of $v_{i,j}$ in $G$
    \EndFor
    \For{$j=1 \ldots |G|$}
      \State Let $g_j$ be $j$-th group
      \State Set rounding precision to $\Theta(\log{b_i})$ most significant bits
      \State Round to $L_j$ the length of segments in $g_j$
      \State Round to $C_{j+1}$ the prefix sum on sizes of groups up to $j$
      \State Increase size of $g_j$ if necessary
      \State Round to $S_{j+1}$ the prefix sum on the length of segments in groups up to $j$ 
    \EndFor
    \State Create the dictionary on $S_j$ values for $u_i$
    \State Create an array with tuples $(S_j, C_j, L_j)$   
  \EndFor
  \State Set $\seglen(u_1)$ to an appropriate upper bound value
  \EndFunction
\end{algorithmic}
\caption{High-level description of the first phase of the encoder for a labeling scheme achieving a constant time query.}
\label{alg:consttime}
\end{algorithm}

Now the assignment of the values in the first phase of the encoder, as in Algorithm~\ref{alg:1phaseagain}, becomes a bit more involved.
In a given node, $\seglen$ values of all light children are gathered, then rounding in classes and groups happens.
Then the groups are considered one by one.
Firstly, we round the length of the reserved segment for a group $g_j$ (this corresponds to $L_j$).
Secondly, the size of a group might be increased, by moving some children from the further groups 
(with increasing sizes of their reserved segments to $L_j$).
This is to ensure that the prefix sum on groups' sizes ($C_{j+1}$) can be stored exactly on small number of bits.
Then a recursive assignment of the values in the subtrees of children from this group happens.
At the end the prefix sum on segments' lengths is rounded (this corresponds to $S_{j+1}$).
Observe that gaps between reserved segments are introduced --- some additional intervals of 
possible values for IDs are skipped when a prefix sum is rounded up.
$S_j$ prefix sums are counting also these gaps, not only 'useful' reserved segments.
A very high-level pseudocode for this phase is presented in Algorithm~\ref{alg:consttime}.
The accumulator in the second phase of the encoder is increased according to the sequence of $S_j$
and light weights of nodes.

\paragraph{Storing prefix sums.}
First, we are going to prove Lemma~\ref{lem:rank2}.

\begin{proof}
(Lemma~\ref{lem:rank2})
We just use the method from the paper of Fredman and Willard about fusion trees~\cite{FusionTrees}.
Recall that we have a collection of $k$ numbers $a_1 \leq \ldots \leq a_k$, each of binary length $w'$ (possibly with leading zeroes).
We store them explicitly separated with ones, as a binary string $s$ of form $s=1 \circ a_1 \circ 1 \circ a_2 \circ \ldots \circ 1 \circ a_k$.
For a query $\rank(x)$, first by computing the most significant bit in constant time we check whether binary length of $x$ is longer than $w'$.
In that case, $x$ is larger than every $a_i$.
Now we assume that $x$ has length of $w'$ bits, possibly with leading zeroes.
We multiply $x$ by $(0^{w'} \circ 1)^k$ to get $(0 \circ x)^k$,
that is $k$ blocks of $x$ separated with additional single zeroes.
Then we subtract this value from $s$.
Observe that this let us subtract $0 \circ x$ from each $1 \circ a_i$, with no interference from carrying.
First bit of block $i$ of length $w'+1$ will be 1 if and only if $x \leq a_i$.
After subtracting, we AND the result with $(1 \circ 0^{w'})^k$ to get only this first bit in each block.
The numbers were sorted, so now we just need to find the first 'interesting' bit which flipped to 0 from 1.
Finding the most significant bit can be implemented with a constant number of standard operations, but
we note that for this monotonic case, equivalently, we need to find the number of 1-bits.
To achieve this in our case, we can multiply by $(0^{w'} \circ 1)^k$ --- all the bits set to 1 will collide in the first block of the result,
without a carry from other blocks for values of $w'$ and $k$ used later, so we can find the sum of ones by looking at this first block.
If we subtract the resulting sum from $k$, we get $\rank(x)$.

Note that value of $(0^{w'} \circ 1)^k$ can be produced in constant time by computing
$(1 \ll (k(w'+1)))/((1 \ll (w'+1))-1)$, using C-like $\ll$ notation for a shift to the left.
\end{proof}

Finally, we should go back to the issue with using a dictionary while having numbers in rounded two-parts representation.
In the following, by light abuse of notation, we use $s$ as either a number or a binary string representing this number,
possibly padded with some leading zeroes.
Assume that we have two numbers $x_1$ and $x_2$ in two-parts representation, with the same number of most significant bits stored.
Let $x_1=m_1 2^{e_1}$ and $x_2=m_2 2^{e_2}$, with $|m_1|=|m_2|$.
Now, if the shorter of binary strings $e_1$ and $e_2$ is padded with leading zeroes to have the same length as the longer one,
we can compare $x_1$ and $x_2$ by comparing concatenated parts of their representation:

\begin{fact}
Assuming $|m_1|=|m_2|$ and $|e_1|=|e_2|$, $x_1 > x_2$ iff $e_1 \circ m_1 > e_2 \circ m_2$, as firstly exponents are compared, and only when they are equal
stored most significant bits are compared.
\end{fact}
Recall that created prefix sums have the same number of stored most significant bits, and we can add leading zeroes to
the exponents to make their lengths equal, which is also necessary for placing them in a dictionary.
Therefore, we can use a dictionary from Lemma~\ref{lem:rank2} to store created prefix sums and answer $\rank$ queries in constant time.
The answer to a query is used as an index in the array storing tuples $(S,C,L)$.

Let us sum up the process of obtaining the sought prefix sums:
\begin{enumerate}
\item The decoder computes $q=\start(w)-\start(u)$ and translates it into two-parts representation,
by changing all bits after fixed number of most significant ones to zeroes,
then constructs a string of fixed length by concatenating these two parts and padding with leading zeroes.
\item The resulting number is used in a $\rank$ query to the dictionary.
\item The acquired rank is used to access an index in an array, where values of $L, C, S$ are stored.
With them, the decoder can answer a query.
\end{enumerate}

All steps can be executed in constant time.
Overall, both the dictionary and the array storing prefix sums take $\Oh(\log{\lightweight(u_i)})$ bits,
with $\Oh(\log{\lightweight(u_i)}/\log{\log{\lightweight(u_i)}})$ elements of size $\Oh(\log{\log{\lightweight(u_i)}})$).
For big enough constant $c$ in the formula for $b$ they can be fit into $\log{\lightweight(u_i)}$ bits of space available in the label of every node $u_i$.
Recall that we use $b=\log{\lightweight(u_i)}/c(\log{\log{\lightweight(u_i)}})^2$, thus increasing constant $c$ decreases number of groups
(there are $\Oh(b\log{\log{\lightweight(u_i)}})$ of them),
which in turn decreases size of structures used, allowing us to fit them in $\log{\lightweight(u_i)}$ bits.
As a side effect, though, $b$ decreases and the rounding becomes more rough, increasing the range of the numbers used for IDs.
Namely, the constant in the exponent in $2^{\Oh((\log{\log{\lightweight(u_i)}})^3)}$ factor increases.
But this affects only the constant in the second-order term of the bound on the length of a label,
and we only use the property that the binary length of all numbers stored in our structures is $\Oh(\log{\lightweight(u_i)})$ anyway.

\end{document}